\documentclass[sn-basic]{sn-jnl}

\usepackage{graphicx}%
\usepackage{multirow}%
\usepackage{amsmath,amssymb,amsfonts}%
\usepackage{amsthm}%
\usepackage{mathrsfs}%
\usepackage[title]{appendix}%
\usepackage{xcolor}%
\usepackage{textcomp}%
\usepackage{booktabs}%
\usepackage{algorithm}%
\usepackage{algorithmic}%
\usepackage{bm}
\renewcommand{\vec}[1]{{\bm{#1}}}
\newtheorem{proposition}{Proposition}
\newtheorem{remark}{Remark}

\begin{document}

\title[Co-clustering by Tri-Factorizing the Regression Coefficient Matrix]{Co-clustering of Response and Covariate Variables by Tri-Factorizing Their Non-negative Regression Coefficient Matrix}

\author*[1]{\fnm{Kenichi} \sur{Satoh} (\href{https://orcid.org/0000-0003-4436-9347}{ORCID: 0000-0003-4436-9347})}
\email{kenichi-satoh@biwako.shiga-u.ac.jp}
\author[2]{\fnm{Yomei} \sur{Tokuda} (\href{https://orcid.org/0000-0002-6943-8794}{ORCID: 0000-0002-6943-8794})}
\email{tokuda@edu.shiga-u.ac.jp}
\affil[1]{\orgdiv{Faculty of Data Science}, \orgname{Shiga University}, \orgaddress{\street{Banba 1-1-1}, \city{Hikone}, \postcode{522-8522}, \state{Shiga}, \country{Japan}}}
\affil[2]{\orgdiv{Faculty of Education}, \orgname{Shiga University}, \orgaddress{\street{Hiratsu 2-5-1}, \city{Otsu}, \postcode{520-0862}, \state{Shiga}, \country{Japan}}}

\abstract{Two-block data---two sets of variables measured on the same individuals, such as microbial taxa and metabolites---raise the question of how \emph{groups} of covariate variables relate to \emph{groups} of response variables. Co-clustering answers this for a single matrix, not for two variable blocks; existing two-block methods either cluster only one side or return signed factors rather than clusters. Starting from the multivariate linear regression $Y_1\approx M Y_2$, we give its non-negative coefficient matrix a tri-factorization $M=X_1\Theta X_2$ (a tri-NMF), so that $X_1$ softly clusters the response variables, $X_2$ the covariate variables, and $\Theta$ is a tested matrix of block correspondences. This makes the method the non-negative member of the reduced-rank regression (RRR) family, expressing RRR's low-rank class in a parts-based basis as NMF relates to PCA; the constraint can only restrict the fit, so predictive accuracy is not the aim; the co-clustering and tested correspondences are. We give multiplicative update rules, choose the two ranks by cross-validation, and develop a conditional Wald test for $\Theta$ applied after basis selection; its size is nominal with fixed bases, conservative after re-estimation, and slightly above nominal under correlated responses, while a non-zero path's \emph{magnitude} stays conditional on the estimated bases. We illustrate the method---a tri-factorized non-negative RRR (NMF-RRR)---on four data sets spanning a permutation structure (Doubs, community ecology), a weak cross-structure under $p>n$ (nutrimouse, nutrigenomics), a pronounced one in a screened microbiome--metabolome study (FRANZOSA, where two microbial groups are jointly associated with each metabolite module), and a classification special case (Wine).}


\maketitle

\section{Introduction}\label{sec1}

Non-negative matrix factorization (NMF) of \citet{lee1999,lee2000} approximates a non-negative data matrix by the product of a basis matrix and a coefficient matrix. Because the coefficients are non-negative and normalize to proportions, they admit a direct \emph{soft-clustering} reading \citep{ding2005}, which has made NMF a standard tool for parts-based representation and exploratory analysis. Inserting a third, middle factor yields the \textbf{non-negative matrix tri-factorization (tri-NMF)} of \citet{ding2006}, $V \approx F S G^{\top}$, which clusters the rows and the columns of a single matrix simultaneously---the canonical algebraic form of \textbf{co-clustering}.

Co-clustering has been developed along several lines, all operating on \emph{one} observed matrix and clustering its two index sets: bipartite spectral co-clustering \citep{dhillon2001}, information-theoretic co-clustering \citep{dhillon2003itcc}, latent block models \citep{govaert2013}, and biclustering of gene-expression data \citep{cheng2000}. In each case the two modes are simply the rows and columns of the same matrix (e.g., documents $\times$ words, or genes $\times$ samples).

In many studies, however, the data arrive as \textbf{two blocks of variables measured on the same individuals}---a block of \emph{covariates} (inputs) and a block of \emph{responses} (outputs)---and the scientific question is how \emph{groups of covariates} relate to \emph{groups of responses}. Hepatic gene expression versus fatty-acid concentrations in nutrigenomics, environmental gradients versus species abundances in community ecology, gut microbial taxa versus metabolite concentrations in microbiome studies, and chemical composition versus optical properties in materials science \citep{tokuda2020glass,tokuda2021glass} are all of this kind. One then wishes to co-cluster the covariate variables and the response variables \textbf{jointly}, while respecting the regression of the responses on the covariates.

Existing methods address parts of this problem but not the whole.
\begin{itemize}
\item \textbf{(i) Single-matrix co-clustering} (above) clusters the two index sets of \emph{one} matrix; it has no notion of a covariate--response relationship, uses no covariates, and cannot predict a new individual.
\item \textbf{(ii) Supervised one-sided feature clustering}---the information bottleneck \citep{tishby1999} and information-theoretic word clustering for text classification \citep{dhillon2003feature}---clusters the \emph{covariate} features so as to preserve information about the response, but clusters only one side.
\item \textbf{(iii) Two-block latent-variable methods} relate covariates and responses through shared directions: canonical correlation analysis \citep[CCA;][]{hotelling1936}, partial least squares \citep[PLS;][]{wold1975}, and reduced-rank regression \citep[RRR;][]{anderson1951,izenman1975,reinsel1998}. These return \emph{signed} factors or loadings rather than clusters, so a downstream clustering step is required to obtain groups; and canonical correlation analysis and unregularized reduced-rank regression, in particular, become ill-posed when the number of variables approaches or exceeds the sample size.
\item \textbf{(iv) NMF with covariates} \citep{satoh2023,satoh2026gcm,satoh2025nmfkc} models the NMF coefficient matrix by a \emph{known} covariate matrix, $Y \approx X \Theta A$, recovering exactly the mean structure of the growth curve model \citep[GCM;][]{potthoff1964} and thereby enabling prediction---but the covariate side is \emph{given}, not clustered.
\item \textbf{(v) Other NMF extensions}---deep NMF and deep matrix factorization \citep{trigeorgis2017}, graph-regularized NMF \citep{cai2011}, and kernel NMF \citep{zhang2006knmf}---enrich the factorization with depth, an auxiliary similarity graph, or a nonlinear feature map; but they are formulated as the decomposition or self-reconstruction of a \emph{single} non-negative matrix, not with a distinct covariate block, and none co-clusters a covariate and a response variable set jointly.
\end{itemize}

Closest in spirit are three lines of work. First, \emph{supervised} and \emph{discriminative} variants of NMF \citep{zafeiriou2006,leeh2010} inject label or discriminant information into the factorization of a \emph{single} data matrix. Second, several methods impose structure on \emph{both} sides of the low-rank coefficient $M$ of a two-block regression: sparse and co-sparse reduced-rank and factor regression \citep{chen2012srrr,bunea2012,chenhuang2012,mishra2017secure}, and the sparse orthogonal factor regression of \citet{uematsu2019sofar}, which reads $M$ as a response--covariate ``association network''. Third, and closest of all, \citet{yu2019bicluster} bi-cluster the regression coefficient matrix itself, grouping both the responses and the predictors. These last two lines come nearest to a joint co-clustering of the two variable sets, but the structure they impose is \emph{signed} sparsity/selection (CCA/PLS/RRR and their sparse variants) or a \emph{hard} fusion clustering \citep{yu2019bicluster}---not a non-negative, normalized soft membership, and none equips the between-group links with significance tests. The nearest \emph{non-negative} neighbour is the non-negative (Poisson) reduced-rank regression of \citet{fitzgerald2022} (nn-PRRR), a rank-$r$ \emph{bi-factorization} $M=UV^{\top}$ that ties both blocks to the \emph{same} $r$ components, clustering both sides only jointly---with no separate numbers of response and covariate clusters and no block-correspondence matrix to test. We instead \emph{tri-factorize} the coefficient, $M=X_1\Theta X_2$ with $X_1,X_2$ normalized non-negative---sharing the three-factor algebra of \citet{ding2006}, but replacing their orthogonality constraint by an $\ell_1$ normalization, so that the factors become soft cluster profiles rather than the hard indicators that exact non-negative orthogonality entails---giving $Q$ response clusters, $R$ covariate clusters, and a $Q\times R$ correspondence $\Theta$ whose entries we estimate \emph{and test}: an inferable co-clustering of the response and covariate variables that neither nn-PRRR nor the structured-coefficient methods above provide. Table~\ref{tab:positioning} (Section~\ref{sec2}) places these methods side by side on the axes that distinguish them. The method is thus best read as a \emph{tri-factorized non-negative reduced-rank regression}, the natural next step in the sequence RRR $\to$ nn-PRRR $\to$ NMF-RRR.

We fill this gap by starting from the multivariate linear regression $Y_1\approx M Y_2$ of the responses on the covariates and giving its coefficient matrix the non-negative, three-factor form
\begin{equation}
Y_1 \;\approx\; X_1\,\Theta\,X_2\,Y_2,
\label{eq:model}
\end{equation}
where the responses $Y_1$ and the covariates $Y_2$ are known, $X_1$ (response basis) and $X_2$ (covariate basis) are non-negative, with each column of $X_1$ and each row of $X_2$ summing to one, and $\Theta$ links them. The induced regression coefficient matrix $M=X_1\Theta X_2$ is then a non-negative tri-factorization, so $X_1$ co-clusters the \textbf{response variables}, $X_2$ co-clusters the \textbf{covariate variables}, and $\Theta$ gives the \textbf{block correspondence}. We thus read the method as a co-clustering of the response and covariate variables obtained by tri-factorizing the non-negative regression coefficient matrix; because, algebraically, $Y_1\approx MY_2$ is a non-negative low-rank regression of $Y_1$ on $Y_2$, its estimator is a \emph{tri-factorized non-negative reduced-rank regression} (NMF-RRR), the name we use for it below. The co-clustering is \emph{supervised} in the weak sense that the bases are fitted by regressing $Y_1$ on $Y_2$---in contrast to clustering a single matrix---not in the sense of optimizing predictive accuracy. It is essential to note that it is the \emph{coefficient} matrix $M$ that is tri-factorized, not the data matrix $Y_1$, and that the clusters are of the covariate and response \emph{variables}, not of the samples---the samples serve only as the bridge over which $M$ is estimated.

\emph{Why this model, and what can only it do?} Among methods that structure the regression coefficient $M$, the non-negative \emph{tri}-factorization is the one that simultaneously delivers all three of: (i) \emph{soft} clusterings of \emph{both} the response and the covariate variables---non-negative cluster profiles (soft memberships read from relative loadings) rather than signed loadings or sparse selection; (ii) \emph{separate} group counts $Q\neq R$, so that one response group may be driven by several covariate groups (cross-structure); and (iii) an estimable, significance-\emph{tested} correspondence $\Theta$ between the two groupings. Non-negativity by itself does not suffice---a bi-factorization $M=UV^{\top}$ (such as nn-PRRR) forces $Q=R$ with no $\Theta$ to test; it is the middle factor $\Theta$ that makes the three-factor form the \emph{minimal} structure achieving a tested, two-sided soft co-clustering (Remark~\ref{rem:reparam}; Table~\ref{tab:positioning}).

The model has two complementary readings. The \textbf{primary} one, which we pursue, is co-clustering: $M=X_1\Theta X_2$ is a non-negative tri-factorization---a soft, non-orthogonal counterpart of the tri-NMF of \citet{ding2006}, applied to the regression coefficient matrix---so $X_1$ and $X_2$ softly cluster the response and covariate variables and $\Theta$ records their $Q\times R$ block correspondence, whose entries we estimate and test (a co-clustering augmented by a group-level network). The \textbf{secondary} reading explains the algebra: dropping non-negativity, $Y_1\approx MY_2$ with $\mathrm{rank}(M)\le\min(Q,R)$ is a reduced-rank regression---specializing, under whitening, to CCA and PLS, and to principal component analysis when $Y_1=Y_2$ \citep{eckart1936}---so the estimator is a \emph{tri-factorized non-negative RRR}: it shares RRR's low-rank class (Proposition~\ref{prop:rrr}) but expresses it in a non-negative, parts-based basis, exactly as NMF relates to PCA (their leading fitted directions nearly coincide when the structure is strong, while secondary directions and the basis differ; Section~\ref{sec5}).

This paper makes the following contributions.
\begin{enumerate}
\item We \textbf{formalize NMF with two variable blocks as a co-clustering of the response and covariate variables}, obtained by tri-factorizing the non-negative regression coefficient matrix of $Y_1\approx MY_2$, clarifying that the tri-factorization acts on the coefficient matrix $M$ (not the data) and that it clusters covariate and response \emph{variables} (not samples).
\item We develop \textbf{inference for the block-correspondence matrix $\Theta$}---the central object that is absent from bi-factorized alternatives---using sample-wise robust (sandwich) standard errors and one-sided significance tests, with a wild bootstrap for interval estimates and as a robustness check. This distinguishes a \emph{cross-structure} (one response group linked to several covariate groups, exposed when $Q<R$) from the permutation structure that a square $\Theta$ tends to impose, turning the co-clustering into a tested group-level network linking covariate and response clusters; since for $Q<R$ the covariate-side split is not uniquely identified (Appendix~\ref{app:ident}), these paths are read conditionally on the selected factorization.
\item We \textbf{compare the method, theoretically and empirically, with reduced-rank regression, tri-NMF, and canonical correlation analysis}, reporting not only the differences but the \emph{similarities} (the shared low-rank subspace), and we delineate \emph{when} the non-negative, supervised formulation is genuinely needed (e.g.\ $p>n$, more variables than samples---written $p\gg n$ when far more---or when within-block and cross-block structure disagree) and when conventional methods suffice.
\item We \textbf{illustrate} the method on four two-block data sets---the Doubs fish--environment data (community ecology), the nutrimouse data (nutrigenomics), the FRANZOSA gut microbiome--metabolome study, originally $p\gg n$ and screened to $30$ variables per block for the analysis reported here, and the Wine data---spanning a permutation structure, a weak cross-structure, a pronounced cross-structure, and the classification special case in which a one-hot label reduces the method to a tested co-clustering of features against classes.
\end{enumerate}

The remainder of the paper is organized as follows. Section~\ref{sec2} introduces the model and its relation to reduced-rank regression, CCA, and tri-NMF. Section~\ref{sec3} gives the multiplicative update rules, and Section~\ref{sec:rank} the choice of the two ranks $(Q,R)$. Section~\ref{sec:inference} develops the inference for $\Theta$. Section~\ref{sec5} presents the data analyses and the systematic comparison with related methods, Section~\ref{sec:sim} validates the inference by a simulation study, and Section~\ref{sec6} concludes.

\section{Model and its relation to established methods}\label{sec2}

This section states the model (Section~\ref{subsec:model}), reads its regression coefficient matrix as a non-negative tri-factorization of bounded rank (Section~\ref{subsec:coef}), and locates it among reduced-rank regression, CCA/PLS, tri-NMF, and the autoencoder and growth-curve special cases (Section~\ref{subsec:relations}, with the positioning summarized in Table~\ref{tab:positioning}); Proposition~\ref{prop:rrr} makes the relation to RRR precise and Remark~\ref{rem:reparam} the necessity of the three-factor form. Section~\ref{subsec:cocluster} gives the co-clustering reading.

\subsection{The model}\label{subsec:model}

Let $Y_1=(\vec y_{1,1},\dots,\vec y_{1,N})=(y^{(1)}_{p,n})_{P_1\times N}$ collect, as columns, the response vectors of $N$ individuals on $P_1$ response variables, and let $Y_2=(y^{(2)}_{p,n})_{P_2\times N}$ collect the corresponding covariate vectors on $P_2$ predictors. Both blocks are taken non-negative; signed variables are made non-negative beforehand by a per-variable min--max transform (Section~\ref{sec5}). This transform is central to the method: it makes sign-free variables non-negative, but a negative covariate--response association can then be encoded only indirectly, through membership in a \emph{different} non-negative group rather than through a negative coefficient---a limitation we return to in Section~\ref{sec6}. Writing $M=X_1\Theta X_2$, model~\eqref{eq:model2} below is the multivariate linear regression $Y_1\approx MY_2$ of the responses on the covariates, with its coefficient matrix constrained to a non-negative tri-factorization. We approximate
\begin{equation}
\mathop{Y_1}_{P_1\times N}\;\approx\;\mathop{X_1}_{P_1\times Q}\,\mathop{\Theta}_{Q\times R}\,\mathop{X_2}_{R\times P_2}\,\mathop{Y_2}_{P_2\times N},
\label{eq:model2}
\end{equation}
with $X_1$ the \emph{response basis} ($Q$ response factors), $X_2$ the \emph{covariate basis} ($R$ covariate factors), and $\Theta$ the parameter (link) matrix. We require $X_1,X_2\ge0$ and remove the scale indeterminacy by normalizing each column of $X_1$ and each row of $X_2$ to sum to one; each column of $X_1$ is then a probability vector over the response variables and each row of $X_2$ a probability vector over the covariate variables, so the two bases are \emph{soft co-clusterings} of the two variable sets. We take $\Theta\ge0$, so that every covariate--response path is a non-negative contribution.

For individual $n$, \eqref{eq:model2} reads $\vec y_{1,n}\approx X_1\Theta X_2\,\vec y_{2,n}$. Writing $\vec s_n=X_2\vec y_{2,n}\in\mathbb{R}^{R}$ for the \emph{covariate scores} (a non-negative, low-dimensional summary of the predictors) and $\vec b_n=\Theta\vec s_n\in\mathbb{R}^{Q}$ for the \emph{response scores}, the model decomposes as encoding $\vec s_n=X_2\vec y_{2,n}$, linking $\vec b_n=\Theta\vec s_n$, and decoding $\vec y_{1,n}\approx X_1\vec b_n$.

\subsection{The coefficient matrix and its rank}\label{subsec:coef}

The covariate-to-response map is the \emph{regression coefficient matrix}
\begin{equation}
M=X_1\Theta X_2\quad(P_1\times P_2),\qquad \mathrm{rank}(M)\le\min(Q,R),
\label{eq:coef}
\end{equation}
so that \eqref{eq:model2} is the low-rank multivariate regression $Y_1\approx MY_2$. Since $\mathrm{rank}(M)\le\min(Q,R)$, the attainable fit is bounded by this rank; in our experiments the fit was essentially determined by $\min(Q,R)$, the larger of $Q,R$ adding resolution on its own side at little cost in fit. The non-negativity of $X_1,X_2$ is what turns the three factors into interpretable soft clusters rather than an arbitrary low-rank product.

\subsection{Relation to established methods}\label{subsec:relations}

Model~\eqref{eq:model2} sits at the intersection of three traditions; the differences are summarized in Table~\ref{tab:positioning}.

\paragraph{Reduced-rank regression, CCA, PLS.}
Dropping non-negativity, $M$ in \eqref{eq:coef} is an arbitrary rank-$r$ matrix with $r=\min(Q,R)$, and $Y_1\approx MY_2$ is reduced-rank regression \citep{anderson1951,izenman1975,reinsel1998}; the three factors then collapse to the left and right singular subspaces of the fitted values. Whitening the two blocks by their covariances turns the same construction into canonical correlation analysis \citep{hotelling1936}, and a covariance-based two-block criterion into partial least squares \citep{wold1975}. The proposed method is thus the \emph{non-negative} member of this family, related to RRR as NMF is to PCA. Proposition~\ref{prop:rrr} makes the relationship precise: the two share the rank class, RRR being its unconstrained optimum and NMF-RRR its non-negatively tri-factorizable member; their fitted values coincide only under a condition, but in practice their subspaces are close when the structure is strong (Section~\ref{sec5}).

\begin{proposition}\label{prop:rrr}
Fix $r=\min(Q,R)$. (i) The NMF-RRR coefficient $M=X_1\Theta X_2$ satisfies $\mathrm{rank}(M)\le r$, the same rank class as the rank-$r$ RRR coefficient. (ii) RRR minimizes $\|Y_1-MY_2\|_F^2$ over \emph{all} $M$ with $\mathrm{rank}(M)\le r$, whereas NMF-RRR minimizes it over the non-negatively tri-factorizable subset $\{X_1\Theta X_2:X_1,X_2\ge0~\text{(normalized)},\,\Theta\ge0\}$; hence $\|Y_1-\hat M_{\mathrm{RRR}}Y_2\|_F\le\|Y_1-\hat M_{\mathrm{NMF-RRR}}Y_2\|_F$. (iii) The two fitted matrices coincide iff some feasible non-negative tri-factorized coefficient $M$ satisfies $MY_2=\hat M_{\mathrm{RRR}}Y_2$; in general they differ---and even their column spaces need not coincide---their proximity being an empirical matter (Section~\ref{sec5}).
\end{proposition}

\noindent Part (ii) explains why RRR always attains the better in-sample fit, and part (iii) why ``shared subspace'' is a property to be checked, not assumed---which we do in Section~\ref{sec5}.

\begin{remark}\label{rem:reparam}
As a matrix class the tri-factorization $M=X_1\Theta X_2$ is \emph{not} more general than a rank-$r$ bi-factorization $M=UV^{\top}$; it is a reparameterization with $r=\min(Q,R)$. Its value is therefore not a larger class but the constrained, normalized, non-negative \emph{structure}: $X_1$ and $X_2$ are soft clusterings of the response and covariate variables and $\Theta$ is their $Q\times R$ correspondence, so the coefficient reads as a co-clustering of the two variable sets---which a (non-negative) bi-factorization $UV^{\top}$, tying both blocks to one shared set of $r$ components, does not deliver. Read the other way round, this is a \emph{necessity} argument: once one asks for a non-negative \emph{soft} co-clustering of \emph{both} variable sets---with possibly different numbers of groups $Q\neq R$ and an estimable correspondence between them---non-negativity alone does not suffice, since a bi-factorization forces a single shared rank ($Q=R$, in one-to-one correspondence; see the ``soft (shared)'' row of Table~\ref{tab:positioning}). The middle factor $\Theta$ is exactly what decouples the two groupings and links them, so the three-factor form $M=X_1\Theta X_2$ is the \emph{minimal} non-negative parameterization of the regression coefficient that yields two-sided soft co-clustering.
\end{remark}

\paragraph{Autoencoder and two-layer special cases.}
When $Y_2=Y_1=Y$, \eqref{eq:model2} becomes $Y\approx X_1\Theta X_2 Y$, a non-negative linear autoencoder whose signed counterpart is PCA \citep{eckart1936}. When $X_2=I_{P_2}$ (so $R=P_2$) and $Y_2=A$ is a known covariate matrix, $M=X_1\Theta$ and $Y\approx X_1\Theta A$ is NMF with covariates \citep{satoh2023,satoh2025nmfkc}, whose mean structure is the growth curve model \citep{potthoff1964}. The present model adds a \emph{learned} covariate basis $X_2$, i.e.\ it additionally clusters the covariates.

\paragraph{tri-NMF.}
The coefficient $M=X_1\Theta X_2$ is a non-negative matrix tri-factorization \citep{ding2006}. Two differences matter. First, Ding's tri-NMF factorizes a single observed non-negative matrix $V$ (clustering its rows and columns), whereas here the tri-factorized object is the \emph{coefficient} $M$, estimated by predicting $Y_1$ from $Y_2$. Second, their bi-orthogonal formulation constrains $F^{\top}F=I$ and $G^{\top}G=I$; since a non-negative matrix with orthonormal columns has at most one non-zero per row, exact orthogonality is a \emph{hard} clustering. We replace that constraint by the $\ell_1$ normalization of Section~\ref{subsec:norm}---unit column sums for $X_1$, unit row sums for $X_2$---which keeps the factors interpretable as cluster profiles while allowing a variable to load on several groups. The present model is in this sense a soft counterpart of the orthogonal tri-factorization, not an instance of it. The proposed method is therefore a \emph{supervised} tri-NMF: it co-clusters the response variables ($X_1$) and the covariate variables ($X_2$) with block correspondence $\Theta$, and---unlike tri-NMF, which only describes the matrix it is given---predicts a new individual through $\hat{\vec y}=X_1\Theta X_2\,\vec a$. This is not tri-NMF with one factor merely held fixed, however: because the tri-factorized object is the coefficient $M$ rather than $Y_1$ itself (the first difference above), the fitting criterion is the regression $Y_1\approx X_1\Theta X_2 Y_2$, not tri-NMF's direct reconstruction of an observed matrix---a different construction, not a constrained instance of the same one. Fixing $Y_2$ as a known covariate is what makes this a regression on a known design, the same device by which the growth curve model \citep{potthoff1964} fixes its design matrix (its mean structure coincides with ours, above), and it is precisely this that admits external covariates, prediction, and the connections to established multivariate models, greatly broadening the method's range of application; the author's NMF-VAR \citep{satoh2026var} and NMF-LAB \citep{satoh2026lab} broaden NMF from this same viewpoint.

\begin{table}[h]
\caption{Positioning of the proposed method against related factorizations (discussed in the text). ``Tests $\Theta$'' = inference on the block-correspondence entries; ``Predicts'' = a new individual's response is obtained from its covariates. In the two clustering columns, \emph{soft}/\emph{hard} denote a non-negative soft / a hard clustering of that side and \emph{soft (shared)} a soft membership tied to one shared rank ($Q=R$, one-to-one); \emph{--} marks no clustering of that side (signed loadings or sparse selection only), and \emph{($A$ given)} a known, unclustered covariate block. The bi-orthogonal tri-NMF of \citet{ding2006} is listed as hard because a non-negative matrix with orthonormal columns has at most one non-zero per row; its multiplicative updates approach that constraint only approximately, so fitted factors are in practice softer than the formulation prescribes. Rows follow \citet{ding2006} (tri-NMF), \citet{fitzgerald2022} (nn-PRRR), \citet{chen2012srrr,bunea2012,chenhuang2012,mishra2017secure,uematsu2019sofar} (sparse/co-sparse RRR, SOFAR), and \citet{yu2019bicluster} (bi-clustering of $M$).}\label{tab:positioning}
\centering
\footnotesize\setlength{\tabcolsep}{3pt}
\begin{tabular}{lcccccc}
\toprule
Method & Object factorized & Non-neg. & Clust.\ covariates & Clust.\ responses & Tests $\Theta$ & Predicts \\
\midrule
tri-NMF (bi-orthogonal)               & single matrix $V$    & yes & hard (orth.) & hard (orth.) & no & no \\
RRR / CCA / PLS                        & $M$ (low rank) & no  & --   & --      & no & yes \\
Sparse/co-sparse RRR, SOFAR           & $M$ (sparse) & no  & --   & --      & no & yes \\
nn-PRRR & $M{=}UV^{\top}$ & yes & soft (shared) & soft (shared) & no & yes \\
Sup.\ bi-clustering of $M$            & $M$ & no  & hard   & hard   & no & yes \\
NMF w/ covariates (GCM)                & $X_1\Theta$    & yes & ($A$ given) & soft ($X_1$) & one-sided & yes \\
\textbf{Proposed (NMF-RRR)}              & $M{=}X_1\Theta X_2$ & yes & soft ($X_2$) & soft ($X_1$) & yes & yes \\
\bottomrule
\end{tabular}
\end{table}

\paragraph{Model depth.}
Three factors are not arbitrary; they are the maximal meaningful depth. Because $\Theta$ carries no constraint beyond (optional) non-negativity, any deeper product with \emph{unconstrained} inner factors collapses---stacking matrices between the response basis $X_1$ and the known $Y_2$ multiplies them into one---so the model reduces to one of only two forms,
\[
Y_1 \approx X_1\,\Theta\,Y_2 \qquad\text{or}\qquad Y_1 \approx X_1\,\Theta\,X_2\,Y_2:
\]
the two-factor form, which puts a response basis on the raw covariates (NMF with covariates, i.e.\ $X_2=I$), or the three-factor form of this paper, which additionally clusters the covariates. No unconstrained product is richer than the three-factor model, so it is the \emph{maximal} structure that softly co-clusters \emph{both} variable sets. The reduction breaks only when the inner factors are themselves constrained (sparsity, orthogonality, or non-negativity), as in deep NMF \citep{trigeorgis2017}.

\subsection{Co-clustering interpretation}\label{subsec:cocluster}

Because $X_1$ and $X_2$ are non-negative, the columns of $X_1$ and the rows of $X_2$ are non-negative cluster \emph{profiles}---the normalized loadings of the response and, respectively, the covariate variables on each group; the soft membership of a variable across the groups is read from its \emph{relative} loadings across these profiles---for response variable $i$, $\pi_{iq}=(X_1)_{iq}/\sum_{h}(X_1)_{ih}$, and analogously for a covariate variable from the columns of $X_2$---giving a soft co-clustering of the two variable sets. The entry $\theta_{q,r}$ measures how strongly covariate group $r$ drives response group $q$. When $Q=R$ the optimum tends towards a near-permutation $\Theta$---a one-to-one correspondence between covariate and response groups---whereas $Q<R$ exposes cross-structure in which one response group integrates several covariate groups. Section~\ref{sec:inference} provides a test that distinguishes significant cross-structure from the permutation pattern that a square $\Theta$ tends to impose.

\section{Multiplicative update rules}\label{sec3}

We estimate the unknown factors of model~\eqref{eq:model2} by minimizing the squared Euclidean (Frobenius) discrepancy
\begin{equation}
D(X_1,\Theta,X_2)=\bigl\|Y_1-X_1\Theta X_2 Y_2\bigr\|_F^2 ,
\label{eq:objective}
\end{equation}
which corresponds to maximum likelihood under a Gaussian error model. Optimizing all factors jointly is non-convex, but, fixing two of the three blocks, $D$ is convex in the remaining one (Section~\ref{subsec:mono}); we therefore minimize $D$ by block coordinate descent, updating each block by a Lee--Seung multiplicative rule \citep{lee2000} that preserves non-negativity and does not increase $D$. Each rule has the form $\xi\leftarrow\xi\odot(\nabla^{-}\oslash\nabla^{+})$, where $-\tfrac12\partial D/\partial\xi=\nabla^{-}-\nabla^{+}$ is the split of the (negative) gradient into its non-negative ``pull'' and ``push'' parts, $\odot$ and $\oslash$ are the elementwise (Hadamard) product and division, and a small $\varepsilon$ is added to the denominator as a numerical safeguard. The monotonicity and stationarity statements below refer to the exact updates ($\varepsilon=0$) with strictly positive denominators; at such a fixed point a strictly positive coordinate has $\nabla^{-}=\nabla^{+}$, so the update leaves it unchanged.

\subsection{Precomputation avoiding the sample dimension}\label{subsec:precompute}

A naive evaluation of \eqref{eq:objective} costs $O(P_1QRP_2N)$ per iteration. Since the sample size $N$ enters only through inner products, we compute, \emph{once}, before the iteration,
\begin{equation}
S=Y_2Y_2^{\top}\in\mathbb{R}_{\ge0}^{P_2\times P_2},\qquad
G_0=Y_1Y_2^{\top}\in\mathbb{R}^{P_1\times P_2},
\label{eq:precompute}
\end{equation}
at cost $O(P_2^2N+P_1P_2N)$; $S\ge0$ because $Y_2\ge0$, while $G_0$ inherits the sign of $Y_1$ (it is non-negative after the per-variable transform of Section~\ref{sec5}). Within each iteration we then form the small matrices
\begin{align}
P&=X_1^{\top}X_1\in\mathbb{R}_{\ge0}^{Q\times Q}, &
S_X&=X_2 S X_2^{\top}\in\mathbb{R}_{\ge0}^{R\times R}, &
G_X&=X_1^{\top}G_0 X_2^{\top}\in\mathbb{R}^{Q\times R}, \notag\\
G_{0X}&=G_0 X_2^{\top}\in\mathbb{R}^{P_1\times R}, &
A&=X_1^{\top}G_0\in\mathbb{R}^{Q\times P_2}, &&
\label{eq:aux}
\end{align}
all independent of $N$. The updates below are written entirely in terms of \eqref{eq:precompute}--\eqref{eq:aux}.

\subsection{Update for the parameter matrix $\Theta$}\label{subsec:upTheta}

Fixing $X_1,X_2$ and writing $\hat Y_1=X_1\Theta(X_2Y_2)$,
\begin{equation}
\frac{\partial D}{\partial\Theta}
=-2X_1^{\top}\bigl(Y_1-X_1\Theta X_2Y_2\bigr)(X_2Y_2)^{\top}
=-2\bigl(G_X-P\,\Theta\,S_X\bigr),
\label{eq:gradTheta}
\end{equation}
using $X_1^{\top}Y_1(X_2Y_2)^{\top}=X_1^{\top}G_0X_2^{\top}=G_X$ and $X_1^{\top}X_1\Theta(X_2Y_2)(X_2Y_2)^{\top}=P\Theta S_X$. With $\Theta\ge0$ and $G_X\ge0$, the pull and push parts are $\nabla^{-}=G_X$ and $\nabla^{+}=P\Theta S_X$, giving
\begin{equation}
\boxed{\;\Theta\leftarrow\Theta\odot\frac{G_X}{P\,\Theta\,S_X+\varepsilon}\;}
\label{eq:upTheta}
\end{equation}

\subsection{Update for the response basis $X_1$}\label{subsec:upX1}

Fixing $\Theta,X_2$ and writing $B=\Theta X_2Y_2$ (the $Q\times N$ response scores), $D=\|Y_1-X_1B\|_F^2$ and $\partial D/\partial X_1=-2(Y_1B^{\top}-X_1BB^{\top})$. Since $Y_1B^{\top}=G_0X_2^{\top}\Theta^{\top}=G_{0X}\Theta^{\top}$ and $BB^{\top}=\Theta S_X\Theta^{\top}$,
\begin{equation}
\boxed{\;X_1\leftarrow X_1\odot\frac{G_{0X}\,\Theta^{\top}}{X_1\,\Theta S_X\Theta^{\top}+\varepsilon}\;}
\label{eq:upX1}
\end{equation}
which coincides with the ordinary NMF basis update with coefficient matrix $B$.

\subsection{Update for the covariate basis $X_2$}\label{subsec:upX2}

Fixing $X_1,\Theta$ and writing $H=X_1\Theta$ (the $P_1\times R$ effective decoder), $D=\|Y_1-HX_2Y_2\|_F^2$ and $\partial D/\partial X_2=-2(H^{\top}Y_1Y_2^{\top}-H^{\top}HX_2S)=-2(\Theta^{\top}A-\Theta^{\top}P\Theta\,X_2S)$, using $H^{\top}G_0=\Theta^{\top}X_1^{\top}G_0=\Theta^{\top}A$ and $H^{\top}H=\Theta^{\top}P\Theta$. Hence
\begin{equation}
\boxed{\;X_2\leftarrow X_2\odot\frac{\Theta^{\top}A}{\Theta^{\top}P\,\Theta\,X_2S+\varepsilon}\;}
\label{eq:upX2}
\end{equation}

\subsection{Normalization}\label{subsec:norm}

The factorization is invariant to $X_1\leftarrow X_1D^{-1}$, $\Theta\leftarrow D\Theta E$, $X_2\leftarrow E^{-1}X_2$ for positive diagonal $D,E$. After each sweep we therefore rescale each column of $X_1$ to sum to one and each row of $X_2$ to sum to one, absorbing the scales into $\Theta$; this fixes the scaling indeterminacy above and makes the columns of $X_1$ and the rows of $X_2$ probability vectors (cluster profiles), leaving $\hat Y_1$ unchanged. It does not by itself make the tri-factorization unique---the remaining non-uniqueness, most visible for $Q<R$ (Section~\ref{subsec:rankguide}), is why the inference of Section~\ref{sec:inference} is conditioned on the estimated bases rather than on a claim that the decomposition is uniquely identified.

The normalization also pins a scalar summary of $\Theta$.

\begin{proposition}\label{prop:budget}
Under the normalization above, the entries of $\Theta$ sum to the grand total of the coefficient matrix:
\begin{equation}
\sum_{q=1}^{Q}\sum_{r=1}^{R}\theta_{qr}=\vec1_{P_1}^{\top}M\,\vec1_{P_2},\qquad M=X_1\Theta X_2 .
\label{eq:budget}
\end{equation}
\end{proposition}

\noindent This is immediate from $\vec1_{P_1}^{\top}X_1=\vec1_{Q}^{\top}$ and $X_2\vec1_{P_2}=\vec1_{R}$, which give $\vec1^{\top}M\vec1=(\vec1^{\top}X_1)\Theta(X_2\vec1)=\vec1_{Q}^{\top}\Theta\vec1_{R}$. The identity is elementary and, on its own, says nothing about how $\Theta$ is estimated. It acquires statistical content only in combination with Appendix~\ref{app:ident}: when $M$ is itself identified---which needs $Y_2$ to have full row rank, and so excludes the $P_2>N$ example of Section~\ref{subsec:nutri}, where the regression pins only $MY_2$---and when its grand total is stable across estimates, \eqref{eq:budget} makes the $QR$ path coefficients share a fixed total, so that an error in one entry is necessarily offset by the others. We stress what this does and does not deliver. It is an algebraic constraint on the estimation errors, and it does not by itself say which paths gain the mass and which lose it; the direction is an empirical matter, reported in Section~\ref{sec:sim}, where the compensation the constraint predicts is also verified numerically.

\subsection{Monotonicity and the algorithm}\label{subsec:mono}

\begin{proposition}\label{prop:convex}
With the other two blocks fixed, $D$ is convex in the remaining block: the Hessians are $2(S_X\otimes P)$ for $\Theta$, $2(BB^{\top}\otimes I_{P_1})$ for $X_1$, and $2(S\otimes H^{\top}H)$ for $X_2$, each a Kronecker product of positive-semidefinite Gram matrices and hence positive semidefinite.
\end{proposition}

\begin{proposition}\label{prop:mono}
Take $\varepsilon=0$ with strictly positive denominators. Then each of \eqref{eq:upTheta}, \eqref{eq:upX1}, \eqref{eq:upX2} does not increase $D$; every strictly positive fixed point satisfies the Karush--Kuhn--Tucker conditions $\xi\ge0$, $\partial D/\partial\xi\ge0$, $\xi\odot\partial D/\partial\xi=0$, and a fixed point with a zero coordinate does so provided that coordinate is not \emph{zero-locked} (i.e.\ its gradient there is non-negative). The $\varepsilon>0$ in \eqref{eq:upTheta}--\eqref{eq:upX2} is a numerical safeguard only.
\end{proposition}

\noindent The monotonicity follows the auxiliary-function argument of \citet{lee2000} applied to each convex subproblem; that a strictly positive fixed point satisfies the Karush--Kuhn--Tucker conditions, and the caveat about zero-locked coordinates, are the stationarity properties of multiplicative updates analyzed by \citet{lin2007}. We omit the details. Algorithm~\ref{alg:nrrr} collects the steps. Initialization uses the three-step NMF scheme of \citet{satoh2025nmfkc}; the convergence tolerance is set tightly (Algorithm~\ref{alg:nrrr}) because, when a link lies on the non-negativity boundary, the multiplicative updates approach it slowly and a looser tolerance can stop prematurely. The $k$-means initialization is run from several restarts ($20$ in our analyses), which we recommend whenever the bases feed inference. We confirmed the stability on the data of Section~\ref{sec5}: with multi-start and a tight tolerance, all $20$ random initializations converged to the \emph{same} solution---the in-sample $R^2$ was constant and both the response and the covariate co-clusterings were reproduced (adjusted Rand index $\ge0.999$ to the default for every example). Single-start fitting (one restart) occasionally settles on a slightly worse local optimum---e.g.\ on nutrimouse the response partition falls to adjusted Rand index $0.90$ and $R^2$ to $0.152$---which multi-start removes; the only residual non-uniqueness is the labelling of the \emph{over-parameterized} covariate side when $Q<R$, exactly the (non-)identifiability of Section~\ref{subsec:rankguide}, which leaves the identified response groups and the fit unaffected. This mirrors the initialization robustness reported for NMF with covariates by \citet{satoh2026gcm}.

\begin{algorithm}[h]
\caption{Multiplicative updates for NMF-RRR}
\label{alg:nrrr}
\begin{algorithmic}[1]
\REQUIRE $Y_1\in\mathbb{R}_{\ge0}^{P_1\times N}$, $Y_2\in\mathbb{R}_{\ge0}^{P_2\times N}$, ranks $(Q,R)$, convergence tolerance $\varepsilon_{\mathrm{tol}}$ (e.g.\ $10^{-8}$)
\STATE Precompute $S=Y_2Y_2^{\top}$ and $G_0=Y_1Y_2^{\top}$; initialize $X_1,\Theta,X_2$ by the three-step scheme of \citet{satoh2025nmfkc}, whose $k$-means bases are taken as the best of $20$ multistarts
\REPEAT
  \STATE Update $X_1$ by \eqref{eq:upX1}; rescale each column of $X_1$ to sum one, absorbing the scale into the rows of $\Theta$
  \STATE Update $\Theta$ by \eqref{eq:upTheta}
  \STATE Update $X_2$ by \eqref{eq:upX2}; rescale each row of $X_2$ to sum one, absorbing the scale into the columns of $\Theta$
\UNTIL{$|D^{(t)}-D^{(t-1)}|/\max(D^{(t-1)},1)<\varepsilon_{\mathrm{tol}}$}
\RETURN $X_1,\Theta,X_2$ and the coefficient $M=X_1\Theta X_2$
\end{algorithmic}
\end{algorithm}

The auxiliary quantities \eqref{eq:aux} are refreshed from the current factors before each block update (Gauss--Seidel), and the column/row normalization is applied immediately after the $X_1$ and $X_2$ updates so that the columns of $X_1$ and the rows of $X_2$ remain probability vectors throughout. The per-iteration cost is dominated by the auxiliary products $G_{0X}=G_0X_2^{\top}$ and $A=X_1^{\top}G_0$, at $O(P_1P_2R)$ and $O(P_1P_2Q)$, together with $S_X=X_2SX_2^{\top}$ and $\Theta^{\top}P\Theta X_2S$, at $O(RP_2^2+R^2P_2+QRP_2)$; the $N$-dependent work occurs only in the one-off precomputation \eqref{eq:precompute}, so for $P_2\ll N$ the method scales like ordinary NMF. When the covariate dimension $P_2$ is itself very large (thousands of variables), the $P_2\times P_2$ matrix $S=Y_2Y_2^{\top}$ dominates both memory and the $O(P_2^2)$ cost; one then either keeps the sample-dimension form of the updates (avoiding $S$) or exploits sparsity or a low-rank/landmark approximation of $Y_2$. In the data of Section~\ref{sec5} we side-step this by pre-screening each block to its most variable high-prevalence variables, which also keeps the co-clustering interpretable.

The method and the cross-validation of Section~\ref{sec:rank} are available in the \texttt{nmfkc} R package \citep{satoh2025nmfkc}.

\section{Choice of the two ranks}\label{sec:rank}

The model has two ranks: $Q$ response factors and $R$ covariate factors. Because the attainable approximation has rank $\min(Q,R)$, the \emph{in-sample} fit is monotone non-decreasing in both $Q$ and $R$ and cannot be used to choose them; a \emph{predictive} criterion is required. We use cross-validation, in two complementary forms.

\subsection{Element-wise cross-validation}\label{subsec:ecv}

To select the ranks for \emph{approximation}, we cross-validate over the entries of $Y_1$ \citep[in the spirit of][]{wold1978,owen2009}. Partition the $P_1N$ entries at random into $V$ folds; for each fold $v$, set the held-out entries to ``missing'' through a $0/1$ weight matrix $W^{(v)}$ and minimize the weighted objective $\|W^{(v)}\odot(Y_1-X_1\Theta X_2Y_2)\|_F^2$, for which the multiplicative rules of Section~\ref{sec3} extend to analogous masked (weighted) updates---the weighting being carried through the sample dimension rather than the precomputation of Section~\ref{subsec:precompute}---as implemented in the \texttt{nmfkc} package. Predicting the held-out entries and accumulating their squared errors over folds yields
\begin{equation}
\sigma(Q,R)=\Bigl(\tfrac{1}{P_1N}\textstyle\sum_{v}\sum_{(p,n)\in\text{fold }v}\bigl(y^{(1)}_{p,n}-\hat y^{(1)}_{p,n}\bigr)^2\Bigr)^{1/2},
\label{eq:ecv}
\end{equation}
and we choose the $(Q,R)$ with the smallest $\sigma$, or the ``elbow'' beyond which $\sigma$ flattens. Where a single automatic rule is needed---as in the simulations of Section~\ref{sec:sim}, which repeat the selection thousands of times---we use the parsimonious one-standard-error rule, applied to the mean squared error $\sigma^2$ of \eqref{eq:ecv} rather than to $\sigma$ itself: among the pairs whose mean squared cross-validation error lies within one standard error of its minimum---that standard error being computed across the $V$ fold-wise mean squared errors of the minimizing pair---take the pair with the smallest $Q+R$, ties broken by the smaller $Q$. This element-wise scheme does not require held-out individuals, uses all $P_1N$ entries, and is well suited to the small samples typical of two-block data. It targets the rank of the \emph{approximation} $M=X_1\Theta X_2$ and follows the bi-cross-validation tradition for component-model rank selection, in which holding out sub-blocks of entries---rather than whole rows or columns---gives a criterion that penalizes an over-large rank and yields a clean minimum \citep{wold1978,bro2008,owen2009}. When out-of-sample \emph{prediction} is the goal, the complementary criterion is the sample-wise scheme of Section~\ref{subsec:scv}: because the response of a held-out individual is predicted from its covariates alone, sample-wise cross-validation is a genuine new-individual criterion here and does not suffer the score-re-estimation dependence that undermines it in unsupervised component models.

\subsection{Sample-wise cross-validation}\label{subsec:scv}

Because the model also \emph{predicts}, the ranks may instead be chosen for out-of-sample prediction. Partition the $N$ individuals into folds; fit $X_1,\Theta,X_2$ on the training individuals; and, for each held-out individual, predict its response from its covariates through the explicit encoder, $\hat{\vec y}_1=X_1\Theta X_2\vec y_2$. The $(Q,R)$ maximizing the held-out $R^2$ targets the \emph{predictive} rank directly and would be the criterion of choice when prediction, rather than description, is the primary goal; the analyses of Section~\ref{sec5} select their ranks by the element-wise scheme above.

\subsection{The asymmetric roles of $Q$ and $R$}\label{subsec:rankguide}

Because $\mathrm{rank}(M)\le\min(Q,R)$ and, in our experiments, the attained fit was essentially determined by it, the larger of the two ranks adds resolution on its own side at little cost in fit. The two ranks are therefore asymmetric, and the choice between $Q<R$ and $Q>R$ is a modelling decision, not a fit issue. The key fact is one of identifiability: in the non-degenerate case $\mathrm{rank}(M)=\min(Q,R)$, the side carrying \emph{more} factors than this rank is over-parameterized, so its basis columns span a lower-dimensional space and cease to be individually identifiable.

\paragraph{$Q<R$ (more covariate than response factors).}
If $\mathrm{rank}(M)=Q$ (the non-degenerate case), the response basis $X_1$ is not rank-redundant: its $Q$ columns span a full $Q$-dimensional space. Under the response-separability condition of Remark~\ref{rem:ident}, $X_1$ and the combined covariate signature $G=\Theta X_2$ are then identifiable up to relabelling, so the response groups---the response types we wish to name---are well defined. The surplus resolution falls on the covariate side, where each row of $\Theta$ (one response group) may load on several of the $R$ covariate groups. This is exactly the cross-structure ``one response group integrates several covariate drivers'' (e.g.\ several environmental gradients jointly shaping one species guild), whose significance is assessed in Section~\ref{sec:inference}; at an adequate sample size it is statistically supported (Section~\ref{sec5}).

\paragraph{$Q>R$ (more response than covariate factors).}
If $\mathrm{rank}(M)=R$ (the non-degenerate case, so $\mathrm{rank}(M)<Q$), the situation is problematic precisely on the side we care about. The $Q$ columns of $X_1$ are forced to span only an $R$-dimensional subspace, so they become collinear and the $Q$ response groups are \emph{not separately identifiable}; equivalently, there are more response groups than the $R$-dimensional covariate bottleneck can independently drive, and the surplus $Q-R$ response factors carry no new structure; their individual paths are then not uniquely interpretable, irrespective of their conditional significance. On the Doubs data of Section~\ref{sec5}, for instance, $Q{=}4,R{=}2$ attains \emph{exactly} the same $R^2=0.435$ as $Q{=}2,R{=}2$: in this example the surplus response factors did not improve the fit (the attainable rank being $\min(Q,R)=2$). Although $Q>R$ can be read as a dual cross-structure (``one covariate driver spreads over several response groups''), this reading is confounded with the response-factor redundancy and is rarely worthwhile.

\paragraph{Recommendation.}
We therefore choose $Q$ near the predictive elbow of $\sigma(Q,R)$ and then take $Q\le R$, so that the response side avoids rank-induced redundancy---separate identifiability additionally requiring the separability of Remark~\ref{rem:ident}---while the covariate side is free to reveal cross-structure. $Q>R$ is justified only when a deliberately finer \emph{description} of the responses is wanted, with the understanding that the surplus response groups are neither separately identifiable nor independently predictable. The criterion $\sigma(Q,R)$ of \eqref{eq:ecv} is evaluated in parallel over the $V$ folds (we use $V=5$).

\begin{remark}[Identifiability of the factors]\label{rem:ident}
The over-parameterization just described concerns the \emph{subspaces} spanned by the bases; a sharper statement holds for the factors themselves. Beyond the scaling removed by the normalization of Section~\ref{subsec:norm}, the tri-factorization is pinned down only under a \emph{separability} (anchor-variable) condition. If each response group owns a variable that loads on it alone and $\mathrm{rank}(M)=Q$, then the response profiles $X_1$ and the covariate signatures $G=\Theta X_2$ are identifiable up to relabelling; if in addition $R=Q$, each covariate group owns such an anchor, and $\Theta$ is nonsingular, then $\Theta$ and $X_2$ are separately identifiable up to independent permutations of the two label sets. When $Q<R$---the regime we recommend for exposing cross-structure---this last step fails: $M$ fixes $X_1$ and the product $\Theta X_2$, but not the split of $G$ into $R$ covariate groups and their correspondence $\Theta$, absent a further geometric condition. The inference of Section~\ref{sec:inference} conditions on the estimated bases, and this conditioning does not itself remove the factorization non-uniqueness: a significant $\theta_{qr}$ is a path of the \emph{selected} factorization, not of a uniquely determined one. This is why we read the covariate-side cross-structure through the tested paths of $\Theta$ and the anchor (driver) variables rather than through the exact membership of every covariate variable. The formal statement and proof are in Appendix~\ref{app:ident}.
\end{remark}

\section{Conditional inference for the parameter matrix}\label{sec:inference}

The entries of $\Theta$ quantify how strongly each covariate group drives each response group, so they are the natural targets of inference. We follow the conditional, growth-curve-model approach of \citet{satoh2026gcm}, extended here to condition on \emph{two} estimated bases. The inference is therefore \emph{conditional} and \emph{post-selection}---not the formal selective-inference construction that conditions on the whole selection event, but a Wald test conditional on the chosen bases. Fixing these data-chosen bases makes $\Theta$ algebraically estimable within the working model (when the design $ZZ^{\top}$ is nonsingular), but does not by itself confer statistical validity; and for $Q<R$ the null $\theta_{qr}=0$ is defined only relative to the selected factorization (Section~\ref{subsec:rankguide}), so ``presence of a path'' is not a decomposition-invariant hypothesis. We therefore read the test as an \emph{exploratory screen for the presence of a path in the selected coordinate system}, not a decomposition-invariant confirmatory inference; its calibration, and its limits, are examined by simulation in Section~\ref{sec:sim}.

The multiplicative updates make no distributional assumption. To attach standard errors we introduce, solely for inference and without altering the point estimate, a Gaussian working model conditional on the optimized bases $\hat X_1,\hat X_2$. Writing the estimated covariate scores $Z=\hat X_2 Y_2\in\mathbb{R}^{R\times N}$,
\begin{equation}
Y_1=\hat X_1\,\Theta\,Z+\mathcal E,\qquad \mathrm{vec}(\mathcal E)\sim N(\vec0,\sigma^2 I_{P_1N}),
\label{eq:workmodel}
\end{equation}
which is the GCM mean structure with design matrix $\hat X_1$ and covariate matrix $Z$. Vectorizing, $\mathrm{vec}(Y_1)=(Z^{\top}\otimes\hat X_1)\,\mathrm{vec}(\Theta)+\mathrm{vec}(\mathcal E)$ is linear in $\mathrm{vec}(\Theta)$, with Fisher information $\mathcal I(\Theta)=\sigma^{-2}(ZZ^{\top}\otimes\hat X_1^{\top}\hat X_1)$. Treating $\hat X_1,\hat X_2$ as fixed,
\begin{equation}
\mathrm{vec}(\hat\Theta)\ \dot\sim\ N\!\Bigl(\mathrm{vec}(\Theta),\ \sigma^2\bigl(ZZ^{\top}\otimes\hat X_1^{\top}\hat X_1\bigr)^{-1}\Bigr),
\qquad
\hat\sigma^2=\frac{\|Y_1-\hat X_1\hat\Theta Z\|_F^2}{P_1N-QR},
\label{eq:thetadist}
\end{equation}
Rather than rely on the working covariance $\mathcal I(\Theta)^{-1}$ of \eqref{eq:thetadist}---the iid Gaussian error of \eqref{eq:workmodel} being only a working assumption---we report \emph{sample-wise robust} (sandwich) standard errors,
\begin{equation}
\widehat{\mathrm{Var}}\bigl(\mathrm{vec}(\hat\Theta)\bigr)
=\hat{\mathcal I}^{-1}\Bigl(\tfrac{N}{N-1}\textstyle\sum_{n=1}^{N}\vec s_n\vec s_n^{\top}\Bigr)\hat{\mathcal I}^{-1},
\qquad
\vec s_n=\mathrm{vec}\bigl(-\hat\sigma^{-2}\,\hat X_1^{\top}\vec r_n\vec z_n^{\top}\bigr),
\label{eq:sandwich}
\end{equation}
where $\vec r_n=\vec y_{1,n}-\hat X_1\hat\Theta\vec z_n$ is the residual of individual $n$, $\vec z_n$ its covariate score, and $\hat{\mathcal I}$ is the information of \eqref{eq:thetadist} evaluated at $\hat\sigma^2$; the standard errors are the square roots of the diagonal of \eqref{eq:sandwich}. Summing the score outer products over individuals treats individuals as independent but leaves the within-individual covariance \emph{across responses} unrestricted, so these standard errors are robust to the cross-response correlation that \eqref{eq:workmodel} ignores, and target the same asymptotic covariance as the model-based $\hat{\mathcal I}^{-1}$ when \eqref{eq:workmodel} is correctly specified.

\paragraph{One-sided boundary test.}
Because each $\theta_{qr}\ge0$, significance is assessed by the one-sided boundary test $H_0:\theta_{qr}=0$ against $H_1:\theta_{qr}>0$ using $z=\hat\theta_{qr}/\mathrm{SE}$, with the correspondingly one-sided interval $[\max(0,\hat\theta_{qr}-z_{\alpha}\mathrm{SE}),\infty)$ rather than a symmetric Wald interval that could fall below zero. A contrast of paths, being unconstrained in sign, takes a two-sided interval.

\paragraph{Wild bootstrap.}
When the information is ill-conditioned---for example when $R$ is large or an effect lies on the non-negativity boundary---we also compute a wild (multiplier) bootstrap \citep{wu1986,liu1988} of the per-sample residual scores ($B=500$ centered exponential multipliers, $w_n=\zeta_n-1$ with $\zeta_n\sim\mathrm{Exp}(1)$; cf.\ \citealp{mammen1993} for a two-point alternative), with non-negative projection of the resampled estimates. Attaching the multiplier to each individual's whole residual score vector, it rests on the same independence-across-individuals assumption as \eqref{eq:sandwich}, and yields bootstrap standard errors and percentile intervals as a robustness check. For numerical stability we invert $\hat{\mathcal I}+10^{-8}I$, falling back on a generalized inverse should that fail. Because the bootstrap re-uses the same inverse, it provides a robustness check on the distributional approximation behind \eqref{eq:thetadist}; it is not a remedy for an ill-conditioned or non-identified information matrix, which only a change of estimand or of design can cure. The $z$ and $p$ values reported in Section~\ref{sec5} use the sandwich standard errors of \eqref{eq:sandwich}; the bootstrap supplies the interval estimates.

\paragraph{Conditional nature.}
The inference is conditional on the data-driven bases $\hat X_1,\hat X_2$, and hence subject to post-selection effects \citep{taylor2015}; relative to the two-block growth curve model the only change is that the covariate $Z=\hat X_2Y_2$ is itself estimated, so the inference conditions on both bases. \citet{satoh2026gcm} reported satisfactory calibration for the \emph{detection} of covariate-effect paths---the contrasts of primary interest---in that setting, whereas their magnitudes (like the overall level parameters, confounded with the estimated scale of the bases) should be read with caution. In our simulations with re-estimated bases (Section~\ref{sec:sim}), false-positive control was retained but became conservative rather than exactly calibrated.

\section{Data analysis}\label{sec5}

We illustrate the method on four two-block data sets, summarized in Table~\ref{tab:catalog}. Three serve as archetypes, chosen because NMF-RRR is useful for a different reason in each: a permutation structure that every method recovers (Doubs), a weak cross-structure under $p>n$ (nutrimouse), and a pronounced, biologically interpretable cross-structure in an originally $p\gg n$ study screened to $30$ variables per block (FRANZOSA gut microbiome--metabolome); a fourth example, the Wine data, illustrates the \emph{classification} special case, in which a one-hot class label makes the response basis $X_1$ the identity and the method reduces to a tested co-clustering of the covariates against the classes. The estimated correspondence matrices $\Theta$ for all four are shown together in Figure~\ref{fig:theta}. Each variable is mapped to $[0,1]$ by a per-variable min--max transform, so that sign-free covariates become non-negative; metabolite intensities are $\log$-transformed beforehand. After these transforms the per-variable values are bounded and roughly symmetric, so the Gaussian (squared-error) loss of \eqref{eq:objective} is a reasonable working model for the data analysed here, the heavier skew and the zeros of the raw counts having been absorbed by the $\log$ and the rescaling (we return to this in Section~\ref{sec6}). The fit is reported by the column-centered $R^2=1-\|Y_1-\hat Y_1\|_F^2/\|Y_1-\bar Y_1\|_F^2$, computed identically for every method compared, together with the mean absolute error (MAE) on the $[0,1]$ scale (Table~\ref{tab:compare}); the two give the same ordering. (In the comparison, RRR carries an intercept while the non-negative NMF-RRR does not, so the reported fit gap is conservative for NMF-RRR.) The two ranks are chosen by element-wise cross-validation (Section~\ref{sec:rank}), and the paths $\Theta$ are tested by the inference of Section~\ref{sec:inference}; in the tables, $p$-values below $0.001$ are reported as ``$<0.001$'' (the exact tiny values overstating the precision available), and estimates are starred for significance at a glance ($^{*}p<0.05$, $^{**}p<0.01$, $^{***}p<0.001$). Since each fitted $\Theta$ entails $QR$ simultaneous path tests, we also checked the conclusions against a within-table Bonferroni correction: all substantive paths in Doubs, FRANZOSA and Wine remain significant, while the single marginal nutrimouse path (Resp1$\leftarrow$Cov1, $p=0.041$) is treated as exploratory. We compare with reduced-rank regression at the matched rank $r=\min(Q,R)$, with the unsupervised tri-NMF of the covariate--response association $Y_1Y_2^{\top}$, and with classical canonical correlation analysis. We do \emph{not} include the closest non-negative competitor, nn-PRRR \citep{fitzgerald2022}, among these empirical comparators: as a Poisson bi-factorization $M=UV^{\top}$ that ties both blocks to a single shared rank (Table~\ref{tab:positioning}), it clusters both sides only jointly, at $Q=R$ in one-to-one correspondence, and produces no separately tested block-correspondence matrix $\Theta$---so it shares no comparable output with NMF-RRR, and a Poisson-versus-Gaussian in-sample fit comparison would not isolate the effect of this structural difference. The contrast with nn-PRRR is therefore structural, as set out in Section~\ref{subsec:relations} and Table~\ref{tab:positioning}, rather than a matter of fit. The inference of Section~\ref{sec:inference} estimates the entries of $\Theta$, but its calibration---examined by the simulation of Section~\ref{sec:sim}---supports the \emph{presence} of a path more than its magnitude; we therefore read the significant $\Theta$ entries primarily as tested present/absent links and interpret their sizes only qualitatively.

\begin{table}[h]
\caption{The four analysis examples. $P_1$ is the number of response variables (block $Y_1$; for Wine the one-hot class label), $P_2$ the number of covariate variables (block $Y_2$), $N$ the sample size; $(Q,R)$ are the response/covariate ranks chosen by element-wise cross-validation (for Wine, $Q$ equals the number of classes and only $R$ is selected by cross-validation); $R^2$ is the in-sample, column-centered fit at $\mathrm{rank}=\min(Q,R)$. ``$\Theta$ structure'' classifies the estimated parameter matrix (Figure~\ref{fig:theta}) as \emph{permutation} (a near one-to-one correspondence, no significant off-diagonal), \emph{weak cross} (one significant cross-path), or \emph{cross} (each response group significantly associated with two or more covariate groups).}\label{tab:catalog}
\centering
\footnotesize
\begin{tabular}{lllrcrl}
\toprule
Data set & Response $Y_1$ ($P_1$) & Covariate $Y_2$ ($P_2$) & $N$ & $(Q,R)$ & $R^2$ & $\Theta$ structure \\
\midrule
Doubs (ecology)            & fish species (27)  & environment (11)  & 30  & (2,2) & 0.44 & permutation \\
Nutrimouse                 & fatty acids (21)   & gene expr.\ (120) & 40  & (2,3) & 0.15 & weak cross \\
FRANZOSA (IBD)             & metabolome (30)    & microbiome (30)   & 220 & (2,4) & 0.12 & cross \\
Wine (chemistry)           & cultivar (3)       & chemistry (13)    & 178 & (3,3) & 0.38 & permutation \\
\bottomrule
\end{tabular}
\end{table}

\begin{figure}[h]
\centering
\includegraphics[width=0.8\textwidth]{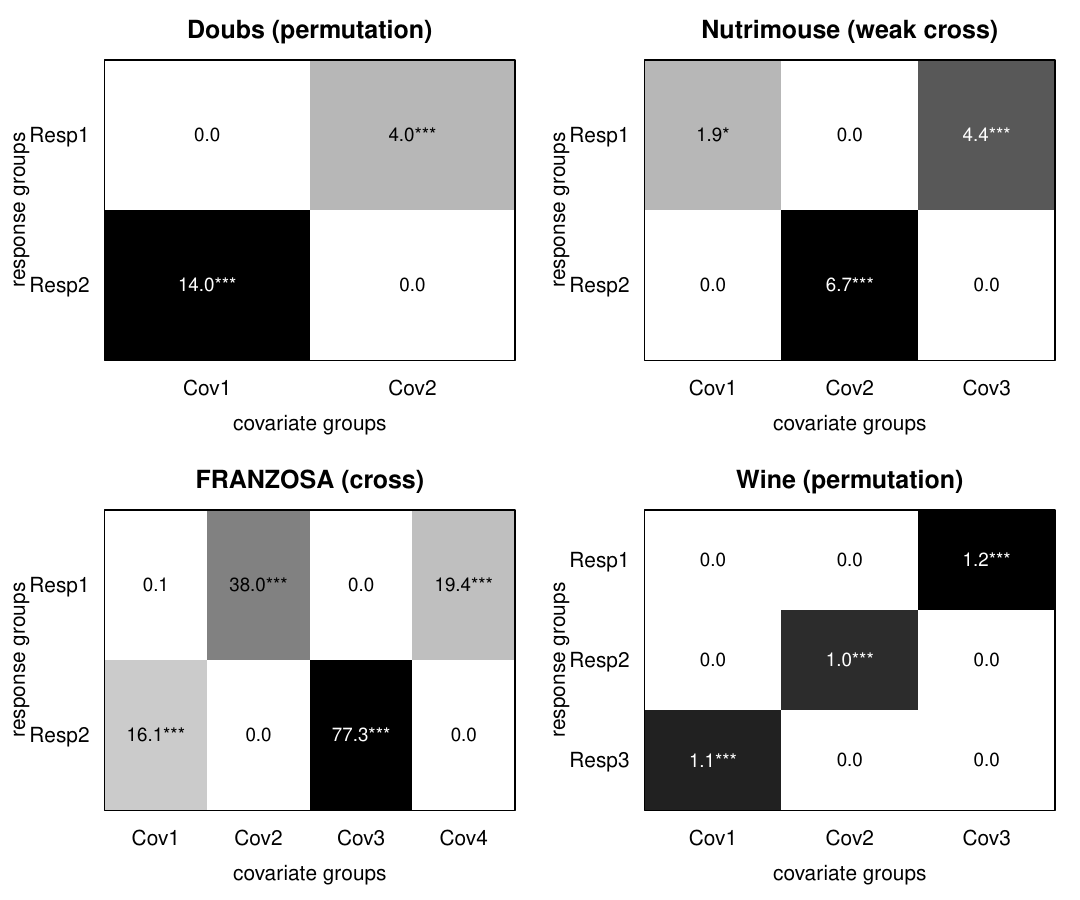}
\caption{Estimated parameter matrix $\Theta$ (the block-correspondence/co-cluster network) for the four examples: rows are response groups, columns covariate groups, cells shaded by magnitude and annotated with the estimate and significance ($^{*}p<0.05$, $^{**}p<0.01$, $^{***}p<0.001$; non-significant cells shown for transparency). Doubs and Wine are near-permutations (one covariate group per response group); FRANZOSA shows a pronounced cross-structure in which each response module is associated with two covariate groups, while nutrimouse shows a weaker cross-structure---one strong cross-path plus a marginal, exploratory one.}\label{fig:theta}
\end{figure}

\subsection{Doubs: fish versus environment}\label{subsec:doubs}

The Doubs data \citep{verneaux1973} record $P_1=27$ fish species and $P_2=11$ environmental variables at $N=30$ sites along a French river; they are a standard illustration of canonical (correspondence) analysis. Element-wise cross-validation selects $Q=R=2$.

The response basis splits the fish into a \textbf{cold-water upstream guild} (Resp1: \emph{Neba}, \emph{Phph}, \emph{Satr}=brown trout, \emph{Cogo}, \emph{Thth}=grayling) and a \textbf{warm-water downstream guild} (Resp2: \emph{Ruru}=roach, \emph{Gogo}, \emph{Baba}=barbel, \emph{Alal}), while the covariate basis splits the environment into an \textbf{oxic upstream gradient} (Cov2: dissolved oxygen and altitude, the two variables with non-negligible loadings) and a \textbf{downstream distance/flow gradient} (Cov1: distance from source $0.567$ and flow $0.433$; nitrate, BOD and the other nutrients load essentially zero, so we do not read Cov1 as a nutrient axis). Table~\ref{tab:doubs} shows that $\Theta$ is a near-permutation: the upstream guild is driven by the oxic gradient and the downstream guild by the distance/flow gradient (both $p<0.001$), the two off-diagonal paths being exactly zero. This reproduces Huet's classical longitudinal zonation, with $R^2=0.435$.

\begin{table}[h]
\caption{Doubs ($Q{=}R{=}2$): inference for $\Theta$ (one-sided boundary test).}\label{tab:doubs}
\centering
\begin{tabular}{llrrrr}
\toprule
Response group & Covariate gradient & Estimate & SE & $z$ & $p$ \\
\midrule
Resp1 (upstream guild)   & Cov2 (oxic)     & $3.97^{***}$  & 0.51 & 7.86 & $<0.001$ \\
Resp2 (downstream guild) & Cov1 (distance/flow) & $14.05^{***}$ & 1.84 & 7.65 & $<0.001$ \\
Resp1 (upstream guild)   & Cov1 (distance/flow) & 0.00  & 0.60 & 0.00 & $0.50$ \\
Resp2 (downstream guild) & Cov2 (oxic)     & 0.00  & 1.00 & 0.00 & $0.50$ \\
\bottomrule
\end{tabular}
\end{table}

The comparison is instructive. Reduced-rank regression fits better ($R^2=0.66$) but its coefficient matrix is signed (about half its entries negative) and yields no clusters. Yet the two methods largely share the same low-rank subspace: their fitted values correlate at $0.73$ and the principal-angle cosines between their response subspaces are $0.99$ and $0.87$, so they differ mainly in the basis of that subspace---signed singular directions versus non-negative parts. The unsupervised tri-NMF of $Y_1Y_2^{\top}$ recovers \emph{exactly} the same guilds and gradients (adjusted Rand index $1.00$ on both sides): the longitudinal gradient is so dominant that the supervised and unsupervised co-clusterings coincide, externally validating the NMF-RRR solution---yet tri-NMF cannot predict the community of a new site. Classical CCA, a factor (not clustering) method, is ill-posed at $N=30$ (canonical correlations $1,1,1,\dots$). Here, then, the value of the proposed method over the simplest baselines is the explicit, testable path structure $\Theta$ rather than the clustering itself.

\subsection{Nutrimouse: gene expression versus fatty acids}\label{subsec:nutri}

The nutrimouse data \citep{martin2007} record hepatic expression of $P_2=120$ genes and concentrations of $P_1=21$ fatty acids in $N=40$ mice---a canonical correlation benchmark with $p>n$. Element-wise cross-validation selects $Q=2$ response and $R=3$ covariate factors.

The fatty acids split into a \textbf{mono-unsaturated group} (Resp1: C16:1, C18:1, C20:1, C14:0) and a \textbf{saturated/long-chain PUFA group} (Resp2: C16:0, C18:0, arachidonic C20:4n-6, DHA C22:6n-3), and the genes into three programs, one of which (Cov3: SR-BI, FAT, Ntcp, S14) is a lipid-transport and metabolism program. The path structure (Table~\ref{tab:nutri}) is now a \emph{non-permutation}: besides the strong Resp2$\leftarrow$Cov2 path, the mono-unsaturated group Resp1 draws on \emph{two} gene programs---strongly on the lipid-handling Cov3 ($p<0.001$) and, more weakly, on Cov1 ($p=0.041$, marginal and not surviving the within-table Bonferroni correction)---a weak, exploratory cross-path that a permutation $\Theta$ could not express ($R^2=0.155$).

\begin{table}[h]
\caption{Nutrimouse ($Q{=}2,R{=}3$): inference for $\Theta$ (one-sided boundary test); non-significant cells omitted.}\label{tab:nutri}
\centering
\begin{tabular}{llrrrr}
\toprule
Response group & Gene program & Estimate & SE & $z$ & $p$ \\
\midrule
Resp2 (sat./PUFA) & Cov2                & $6.75^{***}$ & 0.51 & 13.2 & $<0.001$ \\
Resp1 (MUFA)      & Cov3 (lipid transport/metab.) & $4.39^{***}$ & 0.79 & 5.57 & $<0.001$ \\
Resp1 (MUFA)      & Cov1                & $1.94^{*}$ & 1.12 & 1.74 & $0.041$ \\
\bottomrule
\end{tabular}
\end{table}

Unlike Doubs, the methods now \emph{disagree}, and the regime $p=120>n=40$ is itself informative. Reduced-rank regression attains a higher in-sample fit ($R^2=0.62$, signed), but at $p>n$ this in-sample edge is largely what an unconstrained model with more parameters than observations attains: the regression is under-determined, so the gain need not reflect genuine shared structure or better generalization. (At $p>n$ the RRR normal equations are singular, and we solve them with a minimal ridge---$10^{-3}$ times the mean diagonal of the covariate Gram matrix---purely to make the system solvable; this is numerical stabilization, not a tuned regularizer, and a fully fair predictive comparison would benchmark against a \emph{tuned} ridge or sparse RRR, which we do not pursue here.) Classical CCA, meanwhile, is again ill-posed. The non-negative, normalized, low-rank parameterization of NMF-RRR, by contrast, regularizes the problem through structural constraints---non-negativity and the column/row normalization---rather than through a smaller rank (at the matched rank $r=\min(Q,R)$ the two carry comparable degrees of freedom), so the method remains well-behaved at $p>n$, just as ordinary NMF is routinely fitted to wide matrices. The unsupervised tri-NMF clusters agree only partially with the NMF-RRR ones (adjusted Rand index $0.24$ on the response side and $0.28$ on the covariate side): the NMF-RRR clusters arise from the supervised, low-rank fit and are only partly reproduced by raw association, so the sandwich-based significance of the cross-paths is what licenses their interpretation.

\subsection{FRANZOSA: gut microbiome versus metabolome}\label{subsec:franzosa}

The FRANZOSA inflammatory-bowel-disease study \citep{franzosa2019} measured, in $N=220$ subjects (88 Crohn's, 76 ulcerative colitis, 56 control), a gut microbiome ($11{,}720$ microbial features) and a metabolome ($8{,}848$ LC--MS features)---a strongly $p\gg n$ two-block problem on which classical unregularized CCA and RRR are ill-posed without screening or regularization. For interpretability we keep the $30$ most variable high-prevalence genera and metabolites; microbial relative abundances and $\log$-metabolite intensities are then min--max scaled. Element-wise cross-validation over $Q,R\in\{1,\dots,5\}$ favours $Q\le R$ throughout and never prefers $Q>R$; its global minimum lies at the largest ranks ($Q{=}4,R{=}5$, $\sigma=0.328$), but the profile is essentially flat once a single response factor is present---at $Q{=}2$, $\sigma$ moves only within $0.342$--$0.344$ as $R$ ranges over $2,3,4,5$ (against $0.362$ at $R{=}1$), consistent with the fit being capped by $\mathrm{rank}(M)=Q$ (Section~\ref{subsec:rankguide}). We therefore keep the parsimonious $Q=2$, which holds the two response modules identifiable, and take $R=4$ to resolve the covariate side and expose the cross-structure rather than to lower $\sigma$; like our other modelling choices, this rank is one the subsequent conditional inference conditions on.

Only three of the $30$ retained metabolites carry a compound annotation---the rest are unannotated LC--MS features---so we name each module by its annotated, high-loading members rather than by a metabolite class: Resp1 is the \textbf{urobilin module} (its two annotated features are both urobilin) and Resp2 the \textbf{chenodeoxycholate module} (its top-loading annotated feature is chenodeoxycholate, a bile acid). The four microbial groups, named by their highest-loading genera, are two commensal fibre-fermenter groups (Cov1: \emph{Bifidobacterium}, \emph{Bacteroides}, \emph{Prevotella}; Cov2: \emph{Alistipes}, \emph{Faecalibacterium}), a mixed group led by \emph{Blautia} and also carrying \emph{Coproplasma}, \emph{Enterococcus} and the Proteobacterium \emph{Escherichia} (Cov3), and a \emph{Collinsella}/\emph{Ruminococcus} group (Cov4). The path structure (Table~\ref{tab:franzosa}) is a pronounced \emph{non-permutation}: \emph{each} metabolite module is associated with \emph{two} microbial groups---Resp2 with Cov3 and Cov1, and Resp1 with Cov4 and Cov2---all highly significant. This cross-structure is the statistical finding; the biological reading below attaches to the modules' annotated features and is offered as a hypothesis conditional on the selected factorization, not as a characterization of the modules as chemical classes. Its top-loading feature chenodeoxycholate makes Resp2 a natural place to look for bile-acid biology: microbial transformation of bile acids---bile-salt-hydrolase deconjugation in particular---is a function \emph{distributed} across many genera, including \emph{Bacteroides}, \emph{Bifidobacterium}, \emph{Blautia} and \emph{Enterococcus} \citep{ridlon2014,foley2019}, so a module carrying chenodeoxycholate being jointly associated with a commensal group (Cov1) and a mixed group that includes \emph{Enterococcus} and the Proteobacterium \emph{Escherichia} (Cov3, whose bloom is a hallmark of IBD dysbiosis \citep{shin2015}) is biologically plausible, and chenodeoxycholate is itself among the bile acids that \citet{franzosa2019} report elevated in IBD. We read these as co-occurrence associations, not causal claims, and---since its annotated content is only urobilin---regard Resp1 as more exploratory still. Because $N=220$ is comparatively large, these conditional tests also rest on firmer ground than in the two small benchmarks.

The screening threshold deserves a word. Some filtering is unavoidable---the raw blocks hold $11{,}720$ microbial features and $8{,}848$ metabolite features---and ours is \emph{unsupervised} with respect to the covariate--response relation: prevalence and variance are computed within each block separately and never use the association between the blocks, so this is not selection of variables by their relation to the response. Keeping only $30$ of them is nonetheless an aggressive cut, so we repeated the entire analysis at $50$, $100$, $200$ and $400$ variables per block, holding the ranks at $(Q,R)=(2,4)$ so that the two changes are not confounded. At the fixed ranks the coarse two-path-per-response pattern persisted: at every threshold each metabolite module still drew significantly on exactly two microbial groups, while the fit improved monotonically ($R^2=0.115$, $0.175$, $0.225$, $0.237$, $0.259$). What does change is the \emph{composition} of the covariate groups. As more variables enter, the fine separation among the four covariate groups gives way to coarser ones, until at $400$ variables eight of the eleven genera named above share a single group. The coarse two-path pattern thus persisted across thresholds, whereas the specific group memberships remained threshold-dependent and, like everything on the covariate side, conditional on the $K=30$ screening---an instance, on real data, of the covariate-side non-identifiability that Remark~\ref{rem:ident} flags for $Q<R$. A comparatively large $N=220$ limits sampling variation but does not by itself remove this uncertainty of the factorization. Sample-wise resampling was less stable, particularly on the covariate side; hence the exact microbial memberships and individual paths remain exploratory.

\begin{table}[h]
\caption{FRANZOSA ($Q{=}2,R{=}4$): inference for $\Theta$ (significant paths). Each metabolite module is associated with two microbial groups.}\label{tab:franzosa}
\centering
\begin{tabular}{llrrr}
\toprule
Metabolite module & Microbial group & Estimate & $z$ & $p$ \\
\midrule
Resp2 (chenodeoxycholate)            & Cov3 (\emph{Blautia}/\emph{Coproplasma})   & $77^{***}$ & 20.4 & $<0.001$ \\
Resp1 (urobilin)     & Cov2 (\emph{Alistipes}/\emph{Faecalibacterium}) & $38^{***}$ & 6.4  & $<0.001$ \\
Resp1 (urobilin)     & Cov4 (\emph{Collinsella}/\emph{Ruminococcus})  & $19^{***}$ & 7.3  & $<0.001$ \\
Resp2 (chenodeoxycholate)            & Cov1 (\emph{Bifidobacterium}/\emph{Bacteroides}) & $16^{***}$ & 5.3  & $<0.001$ \\
\bottomrule
\end{tabular}
\end{table}

As in nutrimouse, reduced-rank regression fits better ($R^2=0.30$ versus $0.12$); on the full $11{,}720\times8{,}848$ data CCA and RRR are ill-posed. The unsupervised tri-NMF here reaches substantial agreement with NMF-RRR on the response side (adjusted Rand index $0.74$) and only moderate agreement on the covariate side ($0.45$)---note that these measure agreement between two methods, not accuracy against a known truth. The two metabolite modules are therefore largely visible in the raw association as well; what the supervised formulation adds is the covariate-side partition, which is defined by the regression rather than by the association and is where the two disagree, together with the tested paths of $\Theta$ and the ability to predict a new subject. The relation to RRR is again the one anticipated by Proposition~\ref{prop:rrr}: here the leading principal-angle cosine between the RRR and NMF-RRR response subspaces is $0.97$ while the second is only $0.48$. Indeed, across the three multivariate-response examples the first cosine lies in $[0.90,0.99]$ and the second in $[0.11,0.87]$ (nutrimouse supplies the low ends, $0.90$ and $0.11$)---the two methods agree on the dominant response direction and diverge on the secondary ones, as the proposition predicts when the RRR fit is not itself sign-compatible with a non-negative tri-factorization. This data set is where the cross-structure the method targets is most clearly and most significantly realized.

\begin{table}[h]
\caption{Comparison across the three multivariate-response examples; Wine, whose response is a one-hot class label, is discussed separately in Section~\ref{subsec:wine}. NMF-RRR: proposed (non-negative, supervised). $R^2$ is in-sample, column-centered, at $\mathrm{rank}=\min(Q,R)$; MAE is the mean absolute error on the $[0,1]$ scale. Both metrics use identical definitions for every method, and both rank RRR above NMF-RRR in fit (Proposition~\ref{prop:rrr}). \textsuperscript{a}The FRANZOSA study is originally $p\gg n$ ($11{,}720$ microbial features, $8{,}848$ metabolite features); the analysis reported here screens each block to its $30$ most variable high-prevalence variables, so the fitted problem has $30<220$. Nutrimouse is the one example fitted with more variables than samples.}\label{tab:compare}
\centering
\begin{tabular}{lccc}
\toprule
 & Doubs & Nutrimouse & FRANZOSA \\
 & ($Q{=}R{=}2$) & ($Q{=}2,R{=}3$) & ($Q{=}2,R{=}4$) \\
\midrule
NMF-RRR $R^2$ (non-negative, co-clustering) & 0.44 & 0.15 & 0.12 \\
RRR $R^2$ (signed, no clusters)          & 0.66 & 0.62 & 0.30 \\
NMF-RRR MAE                              & 0.186 & 0.183 & 0.290 \\
RRR MAE                                  & 0.137 & 0.120 & 0.250 \\
tri-NMF vs NMF-RRR, ARI (resp.\ / cov.)          & 1.00 / 1.00 & 0.24 / 0.28 & 0.74 / 0.45 \\
$p$ vs $n$ ; structure                    & $11{<}30$; perm. & $120{>}40$; weak cross & $30{<}220$\textsuperscript{a}; cross \\
\bottomrule
\end{tabular}
\end{table}

Together the three archetypes delineate when the method is needed. When a single dominant gradient aligns both blocks (Doubs), every method agrees and the contribution of NMF-RRR is the testable path structure; when within-block and cross-block structure differ---weakly in nutrimouse and, most clearly, in the FRANZOSA microbiome--metabolome data, where several microbial groups are jointly associated with each metabolite module---the supervised, non-negative formulation and the unsupervised baseline part company: on the covariate side in FRANZOSA, where the response modules are largely shared, and on both sides in nutrimouse. The fourth example moves to a different setting, classification, where the response is a one-hot class label and the response basis therefore degenerates to the identity.

\subsection{Wine: chemical signatures of three cultivars (a classification example)}\label{subsec:wine}

The Wine data \citep{wine1992} record $P_2=13$ chemical measurements (alcohol, phenolics, flavanoids, colour intensity, proline, etc.) on $N=178$ wines from three cultivars grown in the same Italian region. Here the response is the cultivar \emph{label}, so $Y_1$ is the $3\times N$ one-hot indicator and $Y_2$ the chemical covariates. With a one-hot response and $Q=3$, the response basis collapses to the identity (each cultivar is its own group; we verified $X_1=I_3$), so the model specializes to a \emph{supervised co-clustering of the covariates against the classes}---case~(ii) of Section~\ref{sec1} made explicit: the response side is not clustered, and all of the content is in $X_2$ (the grouping of the chemical features) and $\Theta$ (which feature group characterises which cultivar). Element-wise cross-validation selects $R=3$ ($\sigma$ falls from $0.42$ at $R=2$ to $0.37$ at $R=3$ and is flat beyond).

The $13$ features split into three coherent chemical groups, which we name by the features that actually carry their loading: Cov3 (proline and flavanoids), Cov2 (hue and alcalinity of ash), and Cov1 (colour intensity and malic acid). The correspondence $\Theta$ is a clean permutation (Table~\ref{tab:wine}): each cultivar is characterised by exactly one chemical group---Cultivar~1 by the proline/flavanoid group, Cultivar~2 by the hue group, and Cultivar~3 by the colour/malic-acid group---all three diagonal paths highly significant ($p<0.001$) and the six off-diagonal paths exactly zero. The grouping is chemically sensible: Cultivar~1 is the flavanoid- and proline-rich type, whereas Cultivar~3 is set apart by high colour intensity and malic acid. Reduced-rank regression attains a much higher fit ($R^2=0.85$ versus $0.38$; MAE $0.14$ versus $0.34$), as Proposition~\ref{prop:rrr} predicts, but returns signed loadings rather than the parts-based, named feature groups; here the contribution of NMF-RRR is exactly this interpretable, tested feature co-clustering---a permutation $\Theta$ assigning each cultivar its chemical signature. Read as a classifier in the manner of NMF-LAB \citep{satoh2026lab}---with $X_1=I$ the output scores $B=\Theta X_2 Y_2$ become class-membership proportions once each column is normalized to sum one, their calibration as probabilities not being examined here---the hard (argmax) assignment recovers the cultivar for $87.6\%$ of the wines in-sample ($156/178$).

\begin{table}[h]
\caption{Wine ($Q{=}3,R{=}3$): inference for $\Theta$ (one-sided boundary test). Each cultivar is characterised by one chemical feature group; the six off-diagonal paths are exactly zero ($p=0.50$, omitted).}\label{tab:wine}
\centering
\begin{tabular}{llrrr}
\toprule
Cultivar & Chemical feature group & Estimate & $z$ & $p$ \\
\midrule
Cultivar 1 & Cov3 (proline/flavanoids)          & $1.25^{***}$ & 16.0 & $<0.001$ \\
Cultivar 3 & Cov1 (colour intensity/malic acid) & $1.08^{***}$ & 8.9  & $<0.001$ \\
Cultivar 2 & Cov2 (hue/alcalinity of ash)       & $1.04^{***}$ & 7.3  & $<0.001$ \\
\bottomrule
\end{tabular}
\end{table}

This classification example also delimits the method: when the response is a one-hot label and $Q$ equals the number of classes, $X_1$ is the identity and NMF-RRR is no longer a two-block co-clustering but a one-sided, supervised clustering of the features---a useful special case, and a reminder that genuine two-block co-clustering requires either a multivariate response (as in Sections~\ref{subsec:doubs}--\ref{subsec:franzosa}) or $Q$ smaller than the number of classes.

\section{Simulation study}\label{sec:sim}

We run six experiments in two groups. Four concern the inference of Section~\ref{sec:inference}. (i) With the bases held \emph{fixed}, is the test for $\Theta$ well calibrated? (ii) When both bases are \emph{re-estimated} from each data set---the realistic, selective case---do the existence test and the path magnitudes remain reliable? (iii) When the two ranks are \emph{also} chosen from each data set, so that the whole procedure of Sections~\ref{sec3}--\ref{sec:inference} is repeated, does the existence test still control its size? (iv) When the errors are correlated \emph{across responses}, so that the working model \eqref{eq:workmodel} is misspecified, do the sandwich standard errors of \eqref{eq:sandwich} deliver the robustness they are meant to buy? Two further experiments concern the method rather than the test: (v) can it tell a genuine \emph{cross-structure} from a permutation, and (vi) does tri-factorizing the \emph{coefficient} recover the true co-clustering better than factorizing the raw association? The upshot is a clean separation. The size of the existence (significance) test depends on the regime: nominal with the bases fixed, conservative once the bases and---conditional on the true ranks being recovered---the ranks too are re-estimated, and inflated but substantially controlled when responses are correlated, where the sandwich standard errors reduce a five-fold over-rejection to a residual $0.067$--$0.075$. The \emph{magnitudes} of the non-zero paths become anti-conservative as soon as the bases are re-estimated; cross-structure is a detectable feature rather than a dataset artifact; and the co-clustering itself is recovered as well as, and under cross-block covariate correlation more accurately than, by the unsupervised baseline.

The first two experiments are calibrated to the Doubs fit of Section~\ref{subsec:doubs} ($Q{=}R{=}2$). The estimated bases $\hat X_1,\hat X_2$ and parameter matrix are taken as ground truth; the two numerically near-zero entries of $\Theta$ (displayed as $0.00$ in Table~\ref{tab:doubs}) are set exactly to zero, so the design has one genuine path per response group (true values $14.05$ and $3.97$) and two exact boundary nulls. We draw $B=3000$ data sets from the working model $Y_1=\hat X_1\,\Theta\,Z+\mathcal E$ with $Z=\hat X_2Y_2$ and $\mathrm{vec}(\mathcal E)\sim N(\vec0,\hat\sigma^2 I)$, $\hat\sigma=0.243$ being the residual scale of the fit. \emph{Conditional on the bases}---the assumption of Section~\ref{sec:inference}---we estimate $\Theta$, form the model-based standard errors of \eqref{eq:thetadist}, and apply the one-sided boundary test. Because the data are generated from the working model \eqref{eq:workmodel} itself, the model-based and the sandwich standard errors of \eqref{eq:sandwich} target the same quantity here; what the simulation probes is therefore the calibration of the boundary test and the effect of re-estimating the bases, not the robustness that \eqref{eq:sandwich} buys when \eqref{eq:workmodel} is misspecified.

\begin{table}[h]
\caption{Simulation calibrated to the Doubs fit ($Q{=}R{=}2$, $B=3000$, $\hat\sigma=0.243$), conditional on the bases. Empirical coverage of nominal $95\%$ intervals; for the true-zero paths the last column is the size of the one-sided boundary test (target $0.05$), for the others the power.}\label{tab:sim}
\centering
\begin{tabular}{lrrrrr}
\toprule
Path & true $\theta$ & bias & cov.\ (2-sided) & cov.\ (1-sided) & reject ($z>z_{.05}$) \\
\midrule
Resp1$\leftarrow$Cov1 & 0      & $-0.005$& 0.944 & 0.949 & 0.051 \quad(size) \\
Resp2$\leftarrow$Cov1 & 14.05  & 0.003   & 0.957 & 0.951 & 1.000 \quad(power) \\
Resp1$\leftarrow$Cov2 & 3.97   & 0.003   & 0.950 & 0.949 & 1.000 \quad(power) \\
Resp2$\leftarrow$Cov2 & 0      & $-0.008$& 0.951 & 0.945 & 0.055 \quad(size) \\
\bottomrule
\end{tabular}
\end{table}

With the bases fixed (Table~\ref{tab:sim}) the inference is well calibrated: the bias is negligible, both the two-sided and one-sided $95\%$ intervals attain nominal coverage, the boundary test controls its size ($0.051$--$0.055$ against $0.05$), and the power at the two genuine paths is $1$.

In practice the bases are estimated from the same data, which makes the inference selective \citep{taylor2015}. To probe this, we repeat the experiment but \emph{re-estimate both bases} from each simulated data set (clipping negative entries to zero so that $Y_1\ge0$, as the non-negative fit requires) and align the estimated factors to the truth by correlation before applying the inference. Table~\ref{tab:sim2} reports the result.

\begin{table}[h]
\caption{As Table~\ref{tab:sim} but with \emph{both bases re-estimated} from each data set and aligned to the truth ($B=1000$).}\label{tab:sim2}
\centering
\begin{tabular}{lrrrrr}
\toprule
Path & true $\theta$ & bias & cov.\ (2-sided) & cov.\ (1-sided) & reject ($z>z_{.05}$) \\
\midrule
Resp1$\leftarrow$Cov1 & 0      & 0.000   & 1.000 & 1.000 & 0.000 \quad(size) \\
Resp2$\leftarrow$Cov1 & 14.05  & $-1.773$& 0.000 & 0.999 & 1.000 \quad(power) \\
Resp1$\leftarrow$Cov2 & 3.97   & 2.225   & 0.000 & 0.000 & 1.000 \quad(power) \\
Resp2$\leftarrow$Cov2 & 0      & 0.000   & 1.000 & 1.000 & 0.000 \quad(size) \\
\bottomrule
\end{tabular}
\end{table}

A clear dichotomy emerges. The \emph{existence test} stays safe: at the two true-zero paths the boundary test does not over-reject (empirical size $0.000$, i.e.\ conservative in these settings), so a significant path is unlikely to be spurious in the settings examined rather than guaranteed valid in general. The \emph{magnitudes} are not: the two genuine paths are biased by $+2.2$ and $-1.8$, and their nominal two-sided $95\%$ intervals essentially never cover the truth. A one-sided lower interval---natural for a non-negative parameter---does not rescue this; it works for the attenuated path but fails for the over-estimated one, whose lower bound itself exceeds the truth, because the failure is selection bias in the point estimate, not the sidedness of the interval.

The observed pattern is consistent with double-dipping \citep{taylor2015} acting under the fixed total of Proposition~\ref{prop:budget}: the bases are fitted to the same data on which $\Theta$ is then tested, and where the regression pins $\sum_{q,r}\theta_{qr}$ while leaving its allocation across the paths loose, re-estimating the bases \emph{redistributes} that total rather than perturbing the entries independently. Table~\ref{tab:sim2} shows the redistribution directly: the two genuine paths move in opposite directions, by $+2.225$ and $-1.773$, so that the total rises by only $0.452$---some $2.5\%$ of $\sum_{q,r}\theta_{qr}=18.02$---while each individual path is off by four to five times that amount. A dedicated experiment shows that the accounting closes: taking the synthetic $Q{=}R{=}3$ design introduced below, adding a single weak cross path $\theta_{13}=\delta$ and re-running the whole procedure, the deficit at that path and the surplus spread over the three dominant paths differ by $0.02$--$0.03$ for every $\delta$ between $0.1$ and $4$, and that residual is itself the small upward drift of $\sum_{q,r}\theta_{qr}$. The compensation is thus exact up to the drift of the total, which is what \eqref{eq:budget} predicts. The same pattern appeared in the growth-curve study of \citet{satoh2026gcm}.

Two further experiments sharpen the picture. In an auxiliary simulation the upward bias increased with the fraction of latent Gaussian responses clipped at zero; that clipping experiment isolates rectification within the data-generating mechanism, and observed zeros after min--max scaling are not a direct diagnostic of the resulting bias. Together these findings suggest a two-stage mechanism: rectification creates upward pressure on the fitted total, while the normalization identity of Proposition~\ref{prop:budget} constrains how the resulting error is redistributed across the paths. The identity does not determine which paths gain and which lose mass. Neither a full-refit bootstrap (re-estimating the bases in each resample) nor sample splitting (estimating the bases on one half and testing $\Theta$ on the other) restored nominal coverage, so a valid selective-inference correction is left to future work (Section~\ref{sec6}). We therefore read the inference of Section~\ref{sec:inference} as a test of \emph{whether} a path is present---which the simulation supports---and treat the magnitudes of non-zero paths as conditional on the estimated bases.

Experiments (i) and (ii) hold the two ranks fixed at their true values, whereas in practice $(Q,R)$ is chosen from the same data (Section~\ref{sec:rank}), which adds a further layer of selection. The third experiment therefore repeats the \emph{entire} procedure in every replication: data are generated under the null, the two ranks are selected by element-wise cross-validation with the parsimonious one-standard-error rule, the model is re-fitted from twenty multistarts, the factors are aligned, and the sandwich standard errors of \eqref{eq:sandwich} feed the one-sided boundary test. We use two identifiable designs---the Doubs calibration above, and a synthetic block-separable design with $Q{=}R{=}3$, $P_1=15$, $P_2=30$, $N=60$, $\Theta=\mathrm{diag}(10,8,6)$ and $\hat\sigma=0.05$. Identifiability here rests on the conditions of Appendix~\ref{app:ident} and not on $Q=R$ alone: in both designs every response and every covariate group owns an anchor variable loading on it alone (the purest profile is $1.000$ on both sides in both cases), and $\Theta$ is nonsingular and well conditioned---$\det=-55.8$ with $2$-norm condition number $3.5$ for Doubs, $\det=480$ with condition number $1.7$ for the synthetic design. Attention is restricted to $Q=R$ because a per-path null $\theta_{qr}=0$ is not invariant to the choice of factorization when $Q<R$ (Section~\ref{subsec:rankguide}, Appendix~\ref{app:ident}), so the type-I error is well defined only in the identifiable case; sizes are reported conditional on the replication having recovered the true ranks.

\begin{table}[h]
\caption{Full-pipeline null simulation: in every replication the two ranks are re-selected by cross-validation and both bases re-estimated. ``Recovery'' is the proportion of replications selecting the true ranks, on which the remaining columns condition; ``size'' is the largest rejection rate over the true-zero paths (nominal $0.05$) and FWER the probability of rejecting any of them.}\label{tab:sim3}
\centering
\begin{tabular}{lrrrrrr}
\toprule
Design & $B$ & recovery & $n$ & max size & FWER & power \\
\midrule
Doubs, $Q{=}R{=}2$     & 5000 & 1.000 & 5000 & 0.0014 & 0.0014 & 1.000 \\
Synthetic, $Q{=}R{=}3$ & 3000 & 0.746 & 2238 & 0.0000 & 0.0000 & 1.000 \\
\bottomrule
\end{tabular}
\end{table}

Conditional on correct-rank recovery, rank selection did not induce over-rejection (Table~\ref{tab:sim3}): the empirical size never exceeds $0.0014$ against a nominal $0.05$, the family-wise error rate over the true-zero paths is at most $0.0014$ ($0.0002$ after a Bonferroni adjustment), and the power at every genuine path is $1$. Quadrupling the noise of the synthetic design ($\hat\sigma=0.2$) leaves this intact---size $\le0.002$, power $\ge0.998$, with the true ranks recovered in $67\%$ of replications---so the behaviour is not an artefact of an unusually clean design. Two qualifications are essential. First, these are \emph{conditional} rates: in the synthetic design about a quarter of the replications over-select $R$ and are excluded, so Table~\ref{tab:sim3} does not report the unconditional size of the whole procedure. Second, the conditioning is forced rather than chosen for convenience: as noted above, an unconditional per-path type-I error is not even well defined when $Q<R$. What the experiment establishes is thus narrower than it may appear: on the replications where the hypothesis is well posed, adding rank selection on top of basis re-estimation does not create over-rejection. That is consistent with the fixed total of Proposition~\ref{prop:budget}, under which a null path can be declared significant only by taking coefficient mass from the dominant paths.

All three experiments so far draw iid Gaussian errors, that is, exactly the working model \eqref{eq:workmodel}; none of them can therefore speak to the robustness that the sandwich standard errors of \eqref{eq:sandwich} are introduced to buy. The fourth experiment supplies that evidence. Keeping the Doubs calibration and holding the bases \emph{fixed}---so that nothing is confounded by basis re-estimation or rank selection---we draw errors correlated \emph{across responses} within an individual, $\mathcal E_{\cdot n}\sim N(\vec0,\Sigma)$ with $\Sigma_{ij}=\sigma^2\rho^{|i-j|}$ and $\rho\in\{0,0.5,0.8\}$, and compare three standard errors computed on identical data: the model-based $\hat{\mathcal I}^{-1}$ of \eqref{eq:thetadist}, the sandwich of \eqref{eq:sandwich}, and the wild bootstrap of Section~\ref{sec:inference} ($B=2000$ replications, $200$ multiplier draws).

\begin{table}[h]
\caption{Robustness to across-response correlation (Doubs calibration, bases fixed, $B=2000$). Errors within an individual follow $\Sigma_{ij}=\sigma^2\rho^{|i-j|}$; $\rho=0$ is the iid control, where \eqref{eq:workmodel} holds. ``Max size'' is the largest rejection rate over the two true-zero paths (nominal $0.05$); the power at both genuine paths is $1.000$ in every row and is omitted. ``Mean SE'' and ``MC SD'' are averaged over the four paths; the latter is the Monte Carlo standard deviation of $\hat\theta$ itself and so is a property of the design, not of the standard error being compared.}\label{tab:sim4}
\centering
\begin{tabular}{llrrrr}
\toprule
$\rho$ & Standard error & max size & FWER & mean SE & MC SD \\
\midrule
$0$   & model-based    & 0.0525 & 0.0980 & 0.362 & \multirow{3}{*}{0.362} \\
      & sandwich       & 0.0670 & 0.1255 & 0.343 & \\
      & wild bootstrap & 0.0925 & 0.1625 & 0.268 & \\
\midrule
$0.5$ & model-based    & 0.1495 & 0.2615 & 0.360 & \multirow{3}{*}{0.572} \\
      & sandwich       & 0.0700 & 0.1305 & 0.542 & \\
      & wild bootstrap & 0.0875 & 0.1570 & 0.423 & \\
\midrule
$0.8$ & model-based    & 0.2665 & 0.4085 & 0.357 & \multirow{3}{*}{0.830} \\
      & sandwich       & 0.0745 & 0.1330 & 0.792 & \\
      & wild bootstrap & 0.0910 & 0.1620 & 0.620 & \\
\bottomrule
\end{tabular}
\end{table}

Table~\ref{tab:sim4} exposes the mechanism as well as the outcome. At $\rho=0$, where \eqref{eq:workmodel} holds, the model-based test is exactly nominal ($0.0525$). As the correlation grows its size climbs to $0.15$ and then to $0.27$---over five times nominal---because its standard error cannot see the correlation at all: the model-based mean standard error is $0.362$, $0.360$, $0.357$ across the three settings, essentially constant. The sandwich standard error instead follows the true variability: averaged over the four paths it reads $0.343$, $0.542$, $0.792$ against a Monte Carlo standard deviation of $\hat\theta$ of $0.362$, $0.572$, $0.830$, so it reproduces the actual spread to within about five per cent while sitting consistently a little below it. That small downward bias is precisely why the size settles slightly \emph{above} nominal, at $0.067$--$0.075$, rather than at $0.05$; the power at both genuine paths nonetheless remains $1.000$. The wild bootstrap behaves similarly ($0.088$--$0.093$), as one expects since it multiplies each individual's entire score vector and so preserves the within-individual covariance. The protection is not free. At $\rho=0$ the model-based standard error, being correct there, gives the more accurate size ($0.0525$ against $0.0670$), and the bootstrap is the least accurate of the three ($0.0925$) while somewhat under-estimating the standard error. With only $N=30$ individuals available to estimate the $QR\times QR$ meat matrix, this finite-sample cost is expected; we judge it a fair price for the protection obtained at $\rho=0.5$ and $0.8$, and it is the reason the tables of Section~\ref{sec5} report sandwich rather than model-based standard errors.

\subsection*{Detecting cross-structure}\label{subsec:simcross}

A separate concern is whether the cross-structure of Section~\ref{subsec:nutri} is a real, detectable phenomenon or an artifact of one data set. We test this with a synthetic two-block design ($P_1=15$, $P_2=30$, $N=60$, $Q=2$, $R=3$, fixed non-negative bases, $B=2000$, $\hat\sigma=0.05$) under two ground truths: a \emph{permutation} $\Theta=\bigl[\begin{smallmatrix}3&0&0\\0&3&0\end{smallmatrix}\bigr]$ and a \emph{cross} $\Theta=\bigl[\begin{smallmatrix}3&0&2\\0&3&0\end{smallmatrix}\bigr]$ in which response group~1 is driven by covariate groups~1 \emph{and}~3. The one-sided test (conditional on the true bases) gives the rejection rates
\[
\text{permutation:}\ \begin{bmatrix}.85&.05&.05\\.05&.98&.05\end{bmatrix},\qquad
\text{cross:}\ \begin{bmatrix}.85&.05&\mathbf{.64}\\.05&.98&.05\end{bmatrix}.
\]
The cross path (true value $2$) is detected with power $0.64$ and, when absent, rejected at only the nominal $0.05$; the diagonal paths (true value $3$) have power $0.85$ and $0.98$ while every true-zero cell stays near $0.05$. The test thus tells a true cross-structure from a permutation: cross-structure is a detectable feature, not a dataset artifact. We test each $\theta_{qr}$ separately so as to localize \emph{which} covariate group forms the extra path (e.g.\ the chenodeoxycholate module associated with several bacterial groups in Section~\ref{subsec:franzosa}). The power of $0.64$ also reminds us that a non-significant path means ``not detected at this power,'' not evidence of absence.

In a more demanding full-pipeline variant, in which the bases \emph{and} the ranks are re-estimated in every replication, the $\delta$ experiment described above puts a scale on that reminder: with dominant paths of $10$, $8$ and $6$, the weak cross path was declared significant with probability $0.02$ at $\delta=0.5$, $0.24$ at $\delta=1$ and $0.94$ at $\delta=2$, so a cross path much below a fifth of the dominant path in its own row is unlikely to be found. Its magnitude stays strongly shrunk even where detection is reliable---at $\delta=4$ the power is $0.998$ but the mean estimate is $2.6$---which is again consistent with Proposition~\ref{prop:budget}: under a fixed total, the weak path can only be paid for out of the mass held by the dominant ones. This one-fifth heuristic is design-specific rather than universal; power also depends on the sample size, the noise level and the separation of the bases, and the larger $N=220$ of Section~\ref{subsec:franzosa} permits substantially weaker relative paths to be detected.

\subsection*{Recovering the true co-clustering}\label{subsec:simrecovery}

The experiments so far evaluate the \emph{inference}. The object the method is built around, however, is the co-clustering itself, and the case for tri-factorizing the coefficient $M$ rather than the raw association $Y_1Y_2^{\top}$ has so far been algebraic (Section~\ref{subsec:relations}) and indirect---the divergence from unsupervised tri-NMF on the real data of Section~\ref{sec5}, where the truth is unknown. The final experiment tests it where the truth is known.

We reuse the identifiable block-separable design of experiment~(iii) ($Q{=}R{=}3$, $P_1=15$, $P_2=30$, $N=60$, $\Theta=\mathrm{diag}(10,8,6)$) and generate non-negative covariates as $y^{(2)}_{j,n}=0.2+f_{r(j),n}+\lambda\,\mathbf 1\{r(j)\in\{1,2\}\}\,g_n+0.4\,u_{j,n}$, where $f_r$ is a factor common to covariate block $r$, $g$ a factor shared by blocks $1$ and $2$, and $u$ idiosyncratic noise taken large enough that $Y_2$ has full row rank---so that $M$ itself, and not merely $MY_2$, is identified (Appendix~\ref{app:ident}). The parameter $\lambda$ controls how far the association departs from the coefficient. In expectation $Y_1Y_2^{\top}\propto M(Y_2Y_2^{\top})$, so when the covariate Gram matrix is approximately block diagonal \emph{with respect to the true partition}, factorizing the association preserves the co-clustering; cross-block covariance introduces off-block mixing and degrades recovery. We take $\lambda=0$ (mean cross-block covariate correlation $0.03$, within-block $0.86$) and $\lambda=3$ (cross-block $0.27$, within-block $0.94$), with $B=400$ replications each, and compare NMF-RRR against the unsupervised tri-NMF of $Y_1Y_2^{\top}$ used as the baseline in Section~\ref{sec5}, computed with the update rules of \citet{ding2006}, with column normalization to fix the scale they leave free, from a $K$-means start and taking the best of ten restarts.

\begin{table}[h]
\caption{Recovery of the true co-clustering ($Q{=}R{=}3$, $B=400$). The adjusted Rand index is computed against the true blocks; ARI and the two error measures are evaluated at the true ranks $(3,3)$, and only the last row reports the frequency with which cross-validation recovered them. The two error measures refer to NMF-RRR, the only method here that estimates a regression coefficient. $\lambda$ controls the cross-block correlation of the covariates and hence how far the association $Y_1Y_2^{\top}$ departs from the coefficient $M$.}\label{tab:sim5}
\centering
\begin{tabular}{lrr}
\toprule
 & $\lambda=0$ (cross-block corr.\ $0.03$) & $\lambda=3$ (corr.\ $0.27$) \\
\midrule
ARI responses, NMF-RRR            & 0.998 & 1.000 \\
ARI responses, tri-NMF            & 1.000 & 0.823 \\
ARI covariates, NMF-RRR           & 0.995 & 1.000 \\
ARI covariates, tri-NMF           & 1.000 & 0.859 \\
$\|\hat M-M\|_F/\|M\|_F$          & 0.167 & 0.161 \\
$\|\hat MY_2-MY_2\|_F/\|MY_2\|_F$ & 0.010 & 0.003 \\
rank $(3,3)$ recovered            & 0.945 & 0.980 \\
\bottomrule
\end{tabular}
\end{table}

When the covariate Gram matrix is nearly block diagonal with respect to the true partition ($\lambda=0$) the two methods are indistinguishable: both recover the design, the baseline in fact reaching $1.000$ on both sides and agreeing with NMF-RRR in $99\%$ of replications. Nothing is gained by supervising the factorization here---the situation of the Doubs analysis, where the two also agreed exactly. Once the covariate blocks are cross-correlated ($\lambda=3$) the baseline degrades while NMF-RRR does not: the mean adjusted Rand index falls from $1.00$ to $0.823$ on the response side and $0.859$ on the covariate side. The gap is small in the mean but very stable---mean difference $0.177$ and $0.141$ with Monte Carlo standard errors $0.011$ and $0.007$---and it is one-sided: NMF-RRR was never worse than the baseline in any of the $400$ replications, and strictly better in $44\%$ (responses) and $69\%$ (covariates), the remainder being ties. This is the controlled counterpart of the divergence observed on nutrimouse and FRANZOSA, and it locates the advantage precisely: the association matrix displays the coefficient seen through $Y_2Y_2^{\top}$, so supervising the factorization buys nothing when the covariate blocks are uncorrelated and buys accuracy, in this design, when they are not. The design deliberately isolates the distortion contributed by the covariate Gram matrix, and it is favourable to NMF-RRR in other respects---the data are generated from its own model, at the true ranks and at high signal---so it does not establish a universal superiority over tri-NMF. The two error measures also illustrate the identifiability hierarchy of Appendix~\ref{app:ident}: even with $Y_2$ of full row rank, the fitted mean $MY_2$ is recovered an order of magnitude more accurately ($0.003$--$0.010$) than the coefficient $M$ itself ($0.16$).

\section{Discussion and conclusion}\label{sec6}

We introduced a \emph{co-clustering of the response and covariate variables}, obtained by tri-factorizing the non-negative regression coefficient matrix $M=X_1\Theta X_2$ of the multivariate linear regression $Y_1\approx MY_2$: $X_1$ co-clusters the response variables, $X_2$ the covariate variables, and $\Theta$ is an estimated, testable matrix of block correspondences. Because the bases are fitted by regressing $Y_1$ on $Y_2$, the co-clustering is supervised---in contrast to clustering a single matrix---and the estimator is a tri-factorized non-negative reduced-rank regression (NMF-RRR). The construction unifies two traditions that have stood apart: the co-clustering of tri-NMF, which acts on a single matrix and neither uses covariates nor predicts, and the predictive low-rank regression of RRR, CCA, and PLS, which relate two variable blocks but return signed factors rather than clusters.

Three points locate the method. First, dropping non-negativity recovers reduced-rank regression exactly; NMF-RRR and RRR belong to the same low-rank class (Proposition~\ref{prop:rrr}), with RRR its unconstrained optimum and NMF-RRR its non-negative member, and empirically their leading fitted directions nearly coincide while the secondary directions and the basis differ---non-negative parts versus signed singular directions, as NMF relates to PCA. Second, although tri-NMF, with all three factors free, is the most general factorization, fixing the covariate block $Y_2$ as a known covariate is what broadens applicability: it admits external covariates and connects the model to the growth curve model and its relatives. Third, we developed inference for $\Theta$, conditional on the two estimated bases and reported with sample-wise sandwich standard errors---which, as Table~\ref{tab:sim4} shows, substantially limit size inflation when responses are correlated, where the model-based standard error over-rejects by a factor of five---and which distinguishes a cross-structure (one response group linked to several covariate groups, exposed by taking $Q<R$, and read conditionally on the selected factorization since the covariate-side split is then not uniquely identified) from the permutation structure a square $\Theta$ tends to impose, and supports interpretation when the unsupervised baselines disagree.

The three archetypal analyses delineated when the method is needed. On the Doubs data a single dominant gradient aligned both blocks, so unsupervised tri-NMF and NMF-RRR returned the same upstream--downstream guilds; there the contribution of NMF-RRR is the explicit, testable path structure. On the nutrimouse ($p>n$) and, most clearly, the FRANZOSA microbiome--metabolome data (originally $p\gg n$, screened to $30$ variables per block), where within-block and cross-block structure differ, the methods diverged and the supervised non-negative formulation recovered cross-structure---several covariate groups jointly associated with one response module---with the sandwich-based significance of the cross-paths providing conditional support. On FRANZOSA the unsupervised baseline reproduces much of the \emph{response} grouping and the two part company mainly on the covariate side, which the supervised fit defines through the regression rather than through the raw association; on nutrimouse agreement stays low on both sides. The simulation of Table~\ref{tab:sim5} locates when this matters: with the covariate blocks uncorrelated the unsupervised baseline recovers the true co-clustering just as well, so supervising the factorization gains nothing; once they are cross-correlated, so that the association $Y_1Y_2^{\top}$ and the coefficient $M$ carry different block structures, the baseline degrades (mean adjusted Rand index $1.00\to0.82$--$0.86$) while factorizing the coefficient does not---a small but entirely one-sided gap. A fourth example, the Wine data, illustrates the classification special case: with a one-hot class label the response basis becomes the identity and NMF-RRR reduces to a tested, parts-based co-clustering of the chemical features against the cultivars, recovering a clean permutation in which each cultivar carries a distinct chemical signature. Reduced-rank regression remained the better in-sample fit throughout, though at $p>n$ (nutrimouse) that edge is largely a property of the unconstrained fit rather than evidence of better structure---where CCA is outright ill-posed, the non-negative, normalized, low-rank parameterization keeps NMF-RRR well-behaved, just as ordinary NMF is fitted to wide matrices; the value of NMF-RRR is the non-negative, parts-based co-clustering with tested paths, not raw accuracy. The non-negativity constraints make it hard to improve the fit, but the covariate (predictive) structure is still useful for \emph{model selection}: because the fitted model predicts held-out entries and units, the two ranks can be chosen by cross-validation (Section~\ref{sec:rank}), guarding the co-clustering against over-fitting---an option unavailable to a pure single-matrix co-clustering.

Several limitations point to future work. As the simulation of Section~\ref{sec:sim} showed, when the bases are re-estimated---and, conditional on the true ranks being recovered, also when the two ranks are re-selected in every replication---the existence test for a path remained conservative in our simulations, whereas the magnitudes of the non-zero paths became anti-conservative. Proposition~\ref{prop:budget} supplies an algebraic constraint consistent with this pattern: where $M$ is identified and its grand total stable, the normalization ties $\sum_{q,r}\theta_{qr}$ to that total, so the paths share a fixed budget and an error in one is offset by the others. It does not, however, predict which paths gain and which lose, and it is silent in the $P_2>N$ regime where only $MY_2$ is identified; deriving the direction of the compensation, and correcting the path-magnitude intervals for it---for example by selective-inference adjustments---remain open. So does an unconditional treatment of rank selection: Table~\ref{tab:sim3} conditions on recovering the true ranks, because for $Q<R$ the per-path null is not decomposition-invariant. Both blocks were treated under a Gaussian (squared-error) loss after a per-variable min--max transform. This transform is a modelling choice with side effects. It sends each variable's minimum to $0$, so a structural zero (a non-detect) and the smallest observed value are mapped together. For compositional or highly sparse data---abundances that live on a simplex and metabolite features with many zeros---it also distorts the relative geometry, and it can only encode a negative covariate--response relation indirectly, through membership in a \emph{different} non-negative group rather than through a negative coefficient. We used it because it makes sign-free covariates non-negative with minimal assumptions and, after the $\log$ step, leaves variables bounded and roughly symmetric so that the Gaussian loss is adequate (Section~\ref{sec5}); count or compositional responses would nonetheless be better served by the (generalized) Kullback--Leibler objective $D(Y_1\,\|\,\hat Y_1)=\sum_{ij}\bigl(Y_{1,ij}\log\tfrac{Y_{1,ij}}{\hat Y_{1,ij}}-Y_{1,ij}+\hat Y_{1,ij}\bigr)$ on the raw data, for which the multiplicative updates extend directly, or by a compositional (log-ratio) treatment of the covariates. Re-fitting the four analyses under this loss left the significant structure intact---the Doubs paths, the FRANZOSA cross-structure and the Wine permutation all reproduced, only the already-marginal nutrimouse cross path losing significance (cf.\ the fragility of weak paths in Section~\ref{sec:sim}). This objective corresponds to a Poisson working model; a full Poisson (count) treatment is left to future work, noting that at large response means the Poisson and Gaussian models nearly coincide---consistent with the close agreement of the two losses here. Finally, kernelizing the covariate block would lift the linear covariate--response map to a nonlinear one while retaining the non-negative co-clustering of the responses.

\section*{Statements and Declarations}

\noindent\textbf{Funding.} This work was partly supported by the Japan Society for the Promotion of Science (JSPS) KAKENHI Grant Numbers 22K11930, 25K15229, 24K03007, 25H00482, and a research grant from the Fuji Seal Foundation.

\noindent\textbf{Competing Interests.} The authors state that there is no conflict of interest.

\noindent\textbf{Ethical Approval.} This study used only publicly available, de-identified data and involved no new recruitment of participants; no additional ethical approval or informed consent was therefore required.

\noindent\textbf{Author Contributions.} K.S. conceived and developed the method, implemented the software (the \texttt{nmfkc} R package), carried out the analyses, and wrote the manuscript. Y.T., from a glass-materials-engineering perspective centred on composition--property relationships, contributed to the motivation and framing of the two-block (covariate--response) co-clustering problem, advised on the interpretation of the model and its results, and reviewed and edited the manuscript. Both authors read and approved the final manuscript.

\noindent\textbf{Data Availability.} The data sets analysed in this study are publicly available: the FRANZOSA gut microbiome--metabolome data from the \texttt{borenstein-lab/microbiome-metabolome-curated-data} repository; the Doubs and nutrimouse data from the R package \texttt{ade4} and the canonical-correlation literature; and the Wine data from the UCI Machine Learning Repository \citep{wine1992}. Scripts that regenerate every table and figure in this paper, with their captured output, are provided as an online resource accompanying this article (Online Resource~1).

\noindent\textbf{Code Availability.} The R package \texttt{nmfkc} \citep{satoh2025nmfkc} implementing the methods described in this paper is publicly available on CRAN at \url{https://CRAN.R-project.org/package=nmfkc} (\url{https://doi.org/10.32614/CRAN.package.nmfkc}); a package vignette (\emph{Co-clustering two variable blocks with NMF-RRR}) reproduces the method on example data.

\noindent\textbf{Use of AI Tools.} The authors used Anthropic's Claude (Anthropic, Claude Opus~4 family models, \url{https://claude.ai/code}, accessed June--July 2026) as an assistant in the revision of R scripts and for English-language editing and restructuring of this manuscript. All methodological choices, the theoretical results, the numerical experiments, and the data analyses were designed, verified, and approved by the authors, who take full responsibility for the content of the article. No research data, results, figures, or tables were generated by artificial intelligence.

\begin{appendices}
\section{Identifiability of the tri-factorization}\label{app:ident}

We record conditions under which the factors of $M=X_1\Theta X_2$ are identifiable from $M$, at fixed ranks $Q,R$. Write $G:=\Theta X_2$, so $M=X_1G$. Throughout, $X_1\in\mathbb R_{\ge0}^{P_1\times Q}$ has unit column sums and $X_2\in\mathbb R_{\ge0}^{R\times P_2}$ unit row sums (Section~\ref{subsec:norm}), and identifiability is understood up to relabelling of the $Q$ response and $R$ covariate groups. We compare factorizations within the class obeying these constraints together with the separability condition below, which is standard in the identifiability theory of non-negative factorization \citep{donoho2003,arora2012}; without it the comparison admits the usual rotational non-uniqueness of NMF.

\paragraph{Separability (anchor variables).}
$X_1$ is \emph{response-separable} if for each response group $q$ there are a variable $i(q)$ and $c_q>0$ with $X_1[i(q),:]=c_q\,e_q^{\top}$ (an anchor loading on $q$ alone); $X_2$ is \emph{covariate-separable} if for each covariate group $r$ there are a variable $j(r)$ and $d_r>0$ with $X_2[:,j(r)]=d_r\,e_r$.

\begin{proposition}\label{prop:ident}
Let $M=X_1\Theta X_2$ with the normalization above.
\emph{(a)} If $\mathrm{rank}(M)=Q$, $X_1$ is response-separable, and $\Theta$ has no zero row, then $X_1$ and the product $G=\Theta X_2$ are determined by $M$ up to a permutation of the $Q$ response labels.
\emph{(b)} If in addition $R=Q$, $X_2$ is covariate-separable, and $\Theta$ is nonsingular, then $\Theta$ and $X_2$ are determined up to independent permutations of the response and covariate labels.
\emph{(c)} Under the assumptions of \emph{(a)}, if $Q<R$, then $\mathrm{rank}(G)\le Q<R$ and the step in \emph{(b)} no longer applies: $X_1$ and $G$ are still identified as in \emph{(a)}, but the factorization of $G$ into $\Theta$ and $X_2$ need not be uniquely determined by $M$; separate identification requires further assumptions specific to overcomplete non-negative factorizations, such as suitable extreme-ray conditions on the columns of $\Theta$.
\end{proposition}

\noindent\emph{Proof.}
\emph{(a)} Let $\mathrm{cone}(\cdot)$ denote the cone generated by the rows of a matrix. Since $M=X_1G$ with $X_1\ge0$, each row of $M$ is a non-negative combination of rows of $G$, so $\mathrm{cone}(M)\subseteq\mathrm{cone}(G)$. Response-separability gives $M[i(q),:]=c_q\,G[q,:]$ with $c_q>0$, so every row of $G$ lies on a ray of $\mathrm{cone}(M)$ and $\mathrm{cone}(G)\subseteq\mathrm{cone}(M)$; the two cones therefore coincide. From $\mathrm{rank}(M)=Q$ and $M=X_1G$ we get $\mathrm{rank}(G)=Q$, so the rows of $G$ are linearly independent and $\mathrm{cone}(G)$ is simplicial, its $Q$ extreme rays being precisely the rays of its rows. For any admissible factorization $M=\tilde X_1\tilde G$ the same argument gives $\mathrm{cone}(\tilde G)=\mathrm{cone}(M)$; since $\tilde G$ consists of $Q$ linearly independent generators of this same $Q$-dimensional simplicial cone, its rows must generate its $Q$ extreme rays, and as those extreme rays are unique, $\tilde G=D\,P^{\top}G$ for a permutation matrix $P$ and a positive diagonal $D$. Because $G$ has full row rank, $M=X_1G=\tilde X_1\tilde G$ forces $\tilde X_1=X_1PD^{-1}$; equating column sums and using that $X_1,\tilde X_1$ have unit columns gives $D=I_Q$, whence $\tilde X_1=X_1P$ and $\tilde G=P^{\top}G$.

\noindent\emph{(b)} Fixing the response labelling identifies $G$, and $G^{\top}=X_2^{\top}\Theta^{\top}$. With $R=Q$ and $\Theta$ nonsingular, $\mathrm{rank}(G^{\top})=Q$; covariate-separability of $X_2$ is response-separability of $X_2^{\top}$, and unit row sums of $X_2$ are unit column sums of $X_2^{\top}$. Applying (a) to $G^{\top}$ yields $\tilde X_2=S^{\top}X_2$ and $\tilde\Theta=\Theta S$ for a permutation $S$. Together with the response permutation $P$, the admissible factorizations are exactly $\tilde X_1=X_1P$, $\tilde\Theta=P^{\top}\Theta S$, $\tilde X_2=S^{\top}X_2$.

\noindent\emph{(c)} The rank drop $\mathrm{rank}(G)=\mathrm{rank}(\Theta X_2)\le Q<R$ removes the full-rank hypothesis that (a) needs for $G^{\top}=X_2^{\top}\Theta^{\top}$; more than that, separability and normalization alone do \emph{not} pin down the split, as the following example shows. Take $Q=1$, $R=2$, $P_2=3$, with $P_1=1$ and $X_1=(1)$, so that $M=G$ and the hypotheses of (a) hold trivially ($\mathrm{rank}(M)=Q=1$, and the lone response variable anchors its group); set
\[
\Theta=(0.4,\ 0.6),\quad X_2=\begin{pmatrix}0.5&0&0.5\\[2pt]0&\tfrac13&\tfrac23\end{pmatrix};\qquad
\tilde\Theta=(0.5,\ 0.5),\quad \tilde X_2=\begin{pmatrix}0.4&0&0.6\\[2pt]0&0.4&0.6\end{pmatrix}.
\]
Both $X_2$ and $\tilde X_2$ are non-negative with unit row sums and covariate-separable (columns $1$ and $2$ anchor the two groups), yet $\Theta X_2=\tilde\Theta\tilde X_2=(0.2,\ 0.2,\ 0.6)$ while the two factorizations are not related by a permutation. Hence for $Q<R$ the separate identifiability of $\Theta$ and $X_2$ is not guaranteed by separability and normalization; it requires further assumptions specific to overcomplete non-negative factorizations (such as suitable extreme-ray conditions on the columns of $\Theta$), and does not follow merely from the inapplicability of the argument in (a). \hfill$\square$

\medskip
\noindent The separability hypothesis---each group owns a variable loading on it alone---is strong and may fail in practice; for full-rank factorizations, uniqueness without it needs the weaker but more intricate minimum-volume or sufficiently-scattered conditions \citep{laurberg2008,fu2019} (the overcomplete case (c) is different, as noted above). Proposition~\ref{prop:ident} makes precise the hierarchy used in the main text: the fitted mean $MY_2$ is identified from the regression model, and the coefficient $M$ additionally when $Y_2$ has full row rank (or, under random design, when the covariate covariance is positive definite); the response profiles $X_1$ and the covariate signatures $\Theta X_2$ are identified under (a); the individual covariate groups and $\Theta$ only under the balanced condition (b); and for $Q<R$ the covariate split is conditional on the selected factorization, in line with the conditional inference of Section~\ref{sec:inference}.
\end{appendices}

\bibliography{nrrr}

\end{document}